\documentclass[final, journal]{IEEEtran}

\IEEEoverridecommandlockouts

\usepackage{enumitem}
\usepackage{amsmath, amssymb}
\usepackage[utf8]{inputenc}
\usepackage[pdftex]{graphicx}
\graphicspath{{./Figures/}{../Figures/}} 
\usepackage{epstopdf}
\usepackage{float}
\usepackage{subcaption} 
\usepackage{booktabs}
\usepackage{threeparttable}
\usepackage{tablefootnote}
\usepackage{amsthm} 
    \newtheorem{lemma}{Lemma}
    
    \newtheorem{definition}{Definition}
    \newtheorem{property}{Property}
    \newtheorem{remark}{Remark}
    \newtheorem{theorem}{Theorem}
    
    \newtheorem{example}{Example}

\DeclareFontFamily{OT1}{pzc}{} 
\DeclareFontShape{OT1}{pzc}{m}{it}{<-> s * [1.10] pzcmi7t}{}
\DeclareMathAlphabet{\mathpzc}{OT1}{pzc}{m}{it}

\usepackage{fancyhdr}

\def\set#1#2{\{ \, #1 \,:\,#2\,\}} 
\newcommand{\cc}[1]{{\mathcal{#1}}} 
\def\R{ {\rm \,I\!R} } 
\def\N{{\mathbb{Z}}} 
\def\inR#1{\in\R^{#1}} 
\def\vv#1{{ \rm \bf{#1}}} 
\def\Dp#1{\mathcal{D}_{+}^{#1}} 
\def\Sp#1{\mathcal{S}_+^{#1}} 
\def\ri#1{{\rm ri}(#1)} 
\newcommand{\T}{^\top} 
\def\sp#1#2{\langle #1,#2\rangle } 
\def\moveEq#1{{}\mkern#1mu} 
\def\Sum#1#2{\sum\limits_{#1}^{#2}} 
\newcommand{\numeq}[2][=]{\stackrel{\scriptstyle(\mkern-1.5mu#2\mkern-1.5mu)}{#1}} 
\DeclareMathSymbol{\shortminus}{\mathbin}{AMSa}{"39} 

\def\bmat#1{\left[\begin{array}{#1}} 
\def\emat{\end{array}\right]} 
\def\bv {\bmat{c} } 
\def\ev {\emat} 
\newcommand {\bsis} {\left\{ \begin{array} }
\newcommand {\esis} {\end{array}\right.}
\newcommand{\bRlist}{\renewcommand{\labelenumi}{(\roman{enumi})} \begin{enumerate}} 
\newcommand{\eRlist}{\end{enumerate} \renewcommand{\labelenumi}{\arabic{enumi}}} 

\def\bx{\bar{x}}
\def\bu{\bar{u}}

\def\xe{x_e}
\def\xs{x_s}
\def\xc{x_c}
\def\ue{u_e}
\def\us{u_s}
\def\uc{u_c}
\def\ze{z_e}
\def\zej{z_{e (i)}}
\def\zs{z_s}
\def\zsj{z_{s (i)}}
\def\zc{z_c}
\def\zcj{z_{c (i)}}

\newcommand{\zLB}{z_m}
\newcommand{\zLBj}{z_{m (i)}}
\newcommand{\zUB}{z_M}
\newcommand{\zUBj}{z_{M (i)}}
\newcommand{\hzLB}{\hat{z}_m}
\newcommand{\hzLBj}{\hat{z}_{m (i)}}
\newcommand{\hzUB}{\hat{z}_M}
\newcommand{\hzUBj}{\hat{z}_{M (i)}}

\def\xh{x_h} 
\def\xhj{x_{hj}} 
\def\xhcero{x_{h0}} 
\def\xH{\vv{x}_H} 
\def\uh{u_h} 
\def\uhj{u_{hj}} 
\def\uhcero{u_{h0}} 
\def\uH{\vv{u}_H} 
\def\xHo{\vv{x}_H^\circ} 
\def\uHo{\vv{u}_H^\circ} 
\def\xeo{x_e^\circ} 
\def\ueo{u_e^\circ} 

\def\costFunH{{H}} 
\def\stageCostH{{\ell_h}} 
\def\offsetCostH{V_h} 

\def\lT{\mathbb{T}} 
\def\lH{\mathbb{H}} 

\def\zb{\mathpzc{z}} 

\begin{document}
\pagestyle{fancy}

\title{\LARGE \bf
    Harmonic based model predictive control for set-point tracking
}

\author{Pablo~Krupa,~Daniel~Limon,~Teodoro~Alamo%
    \thanks{The authors are with the Department of Systems Engineering and Automation, University of Seville, 41092 Seville, Spain (e-mails: \texttt{pkrupa@us.es}, \texttt{dml@us.es}, \texttt{talamo@us.es}). Corresponding author: Pablo Krupa.}%
\thanks{This work was supported in part by MINERCO-Spain and FEDER funds under grants DPI2016-76493-C3-1-R and PID2019-106212RB-C41; and in part by MCIU-Spain and FSE under grant FPI-2017.}%
\thanks{Accepted version of the article published in IEEE Transactions on Automatic Control. DOI:~10.1109/TAC.2020.3047579~\hfill}
\thanks{\copyright{}2020 IEEE. Personal use of this material is permitted.  Permission from IEEE must be obtained for all other uses, in any current or future media, including reprinting/republishing this material for advertising or promotional purposes, creating new collective works, for resale or redistribution to servers or lists, or reuse of any copyrighted component of this work in other works.~\hfill}%
}

\maketitle
\thispagestyle{fancy}


\begin{abstract}
This paper presents a novel model predictive control (MPC) formulation for set-point tracking. Stabilizing predictive controllers based on terminal ingredients may exhibit stability and feasibility issues in the event of a reference change for small to moderate prediction horizons. In the MPC for tracking formulation, these issues are solved by the addition of an artificial equilibrium point as a new decision variable, providing a significantly enlarged domain of attraction and guaranteeing recursive feasibility for any reference change. However, it may suffer from performance issues if the prediction horizon is not large enough. This paper presents an extension of this formulation where a harmonic artificial reference is used in place of the equilibrium point. The proposed formulation achieves even greater domains of attraction and can significantly outperform other MPC formulations when the prediction horizon is small. We prove the asymptotic stability and recursive feasibility of the proposed controller, as well as provide guidelines for the design of its main ingredients. Finally, we highlight its advantages with a case study of a ball and~plate~system.
\end{abstract}

\section{Introduction} \label{sec:Introduction}

Model predictive control (MPC) is an advanced control strategy that is very prevalent in the current literature due to its inherent capability of providing constraint satisfaction while ensuring asymptotic stability of the target equilibrium point. In MPC, the control law is derived from an optimization problem in which a prediction model is used to predict the future evolution of the system over a prediction horizon \cite{Camacho_S_2013}.

In order to provide asymptotic stability of the closed-loop system, two ingredients are typically added to the MPC formulation: the \textit{terminal cost}, which penalizes a certain measure of discrepancy between the reference and the \textit{terminal state} (i.e. the predicted state at the end of the prediction horizon); and the \textit{terminal set}, which is computed as a positive invariant set of the closed-loop system for the given reference \cite{Rawlings_MPC_2017}. Stability is ensured by imposing the terminal state to lie within the terminal set by the addition of a \textit{terminal constraint} to the MPC formulation.

The use of a terminal set and terminal constraint leads to two downsides when the reference to be tracked can change online. The first issue is that the terminal set must be recomputed for the reference every time it changes. If there are a known-before-hand, finite number of references, then a terminal set can be computed offline for each one of them. Otherwise, it must be computed online each time the reference changes, which is typically very computationally demanding. The second issue is that the feasibility of the MPC problem can be lost in the event of a reference change, i.e. there may not be a feasible solution of the MPC optimization problem for the current state and the new reference. This issue is related to the domain of attraction of the MPC controller, i.e. the set of states for which the closed-loop system is asymptotically stabilizable, since the feasibility is lost when the initial state is out of the domain of attraction of the MPC controller for the new reference. The terminal constraint is the main contributor of this issue when the prediction horizon is not large enough. To see this, note that the predicted state must be able to reach the terminal set within the prediction horizon window and that systems are typically subject to input constraints.

These issues are of particular relevance when dealing with the online implementation of MPC in embedded systems. The severely limited computational and memory resources of these systems make them unsuitable for large prediction horizons and for the computation of positive invariant sets online. Possible solutions to mitigate this are to use \textit{explicit} MPC \cite{Bemporad_explicit_2019, Zeilinger_TAC_2011} or to avoid the computation of a positive invariant set by using a singleton as the terminal set as in \cite{Krupa_TCST_20}. However, the former approach does not scale well with the dimension of the system, and the latter may require a prohibitively large value of the prediction horizon in order to provide good closed-loop performance and not suffer a loss of feasibility in the event of reference changes. There are plenty of other publications on the implementation of MPC in embedded systems, e.g. \cite{Huyck_MED_2012, Hartley_IETCST_2014, Shukla_SD_2017, Lucia_IETII_2018, Jerez_IETAC_2014} and \cite{Wang_TCST_2010}. However, the issues that arise when dealing with small prediction horizons, the recursive feasibility of the MPC controller, or the issue of the online computation of terminal sets are rarely discussed in detail in this particular~field.

Another possible approach would be to use a formulation such as the \textit{MPC for tracking} (MPCT) \cite{Limon_A_2008, Ferramosca_A_2009}, which incorporates a steady state artificial reference into the optimization problem as a decision variable. This formulation offers a significant increase of the domain of attraction when compared to standard MPC formulations and only requires the computation of a single terminal set, valid for all references. Additionally, the asymptotic stability and recursive feasibility of the controller is guaranteed, even in the event of a sudden reference change \cite{Limon_TAC_2018}. However, as we show in Section \ref{sec:HMPC:performance}, the closed-loop performance of the controller can suffer in certain systems if the prediction horizon is too small.

In this paper we present an MPC formulation which we call \textit{harmonic based model predictive control for tracking} and label by \textit{HMPC}. This formulation, which was initially introduced in \cite{Krupa_CDC_19}, is of particular interest when dealing with short prediction horizons, as might be the case when working with embedded systems. As shown in the preliminary results \cite{Krupa_CDC_19}, it attains even greater domains of attraction than MPCT or other standard MPC controllers. Additionally, as we show and discuss in this paper by means of a case study using a ball and plate system, the HMPC controller can show a significant performance improvement when the prediction horizon is small. The improvement can be particularly significant for systems with integrator states and/or systems subject to slew rate constraints on its inputs, as is often the case with robotic and mechatronic systems.

The idea behind this formulation is to substitute the artificial reference of the MPCT formulation by an \textit{artificial harmonic reference}, i.e. a periodic reference signal that is composed of a sine term, a cosine term and a constant. The inclusion of this artificial harmonic reference is heavily influenced by the extensions of the MPCT formulation to tracking periodic references \cite{Limon_MPCTP_2016, Kohler_NMPC_18}. However, in this case, even though the artificial reference used is periodic, the reference to be tracked is a (piecewise) constant set-point, instead of a periodic one.

Periodic MPC for tracking formulations, such as the ones cited in the previous paragraph, or, for instance, \cite{Broomhead_A_2015}, are used to track a generic periodic reference signal with period $T$ by using an artificial reference which is also defined as a generic periodic signal of period $T$, during which the system dynamics and constraints must be imposed. This may lead to a large number of constraints in the optimization problem if the period $T$ is long.
These formulations can also be used for tracking constant set-points, where in this case, the period $T$ can be selected as any positive integer. Choosing $T$ greater than one sample time would only be of interest if it provides an improvement in the performance and/or domain of attraction of the controller. However, doing so would result in an increase in the number of constraints of the optimization problem.

The motivation behind the use of a \textit{harmonic} artificial reference, is that \textit{(i)} indeed, the performance and/or domain of attraction when tracking a constant set-point can be improved by using a periodic artificial reference with a period larger than one sample time, as we show in this paper, and \textit{(ii)} that the system dynamics and constraints can be imposed on it by means of the addition of a small amount of constraints to the optimization problem, the number of which does not depend on the prediction horizon of the controller nor on the period of the harmonic signal.
Therefore, its period can be selected to attain the desired performance and/or domain of attraction without affecting the complexity of the optimization problem.

The addition of these constraints leads to the control law of the HMPC controller being derived from the solution of a second order cone programming problem, instead of the more common quadratic programming problems typically found in other MPC formulations. However, this class of convex optimization problem is common in the literature and can be solved by several efficient algorithms \cite{Domahidi_ECOS_2013, Garstka_COSMO_ECC_2019}. In particular, we use the solver COSMO \cite{Garstka_COSMO_ECC_2019}.

A key property of the HMPC controller is that it retains the recursive feasibility and asymptotic stability features of the MPCT formulation, even in the event of reference changes, as we formally prove in this paper. Moreover, as is also the case with certain versions of the MPCT formulation (in particular the one we highlight in this manuscript), it does not require the computation of a terminal set nor terminal cost.

This paper extends the results of \cite{Krupa_CDC_19} by showing the performance advantages of the HMPC controller, formally proving its asymptotic stability and by including the proof of its recursive feasibility. Additionally, we provide some guidelines for the design of one of its main ingredients: the frequency of the artificial harmonic reference.

The paper is organized as follows. Section \ref{sec:Problem:Formulation} describes the class of system under consideration and control objective. The MPCT controller is described in Section \ref{sec:MPCT}. The proposed HMPC controller is presented in Section \ref{sec:HMPC},  with the theorems stating its recursive feasibility and asymptotic stability. A comparison of the closed-loop performance of these two controllers is presented in Section \ref{sec:HMPC:performance}. Guidelines for the selection of the frequency of the artificial harmonic reference are shown in Section \ref{sec:selection:w}. Finally, conclusions are drawn in Section \ref{sec:conclusions}.

\subsubsection*{Notation}
The relative interior of a set $\cc{X}$ is denoted by $\ri{\cc{X}}$.
The set of integer numbers is denoted by $\N$. Given two integers $i$ and $j$ with ${j \geq i}$, $\N_i^j$ denotes the set of integer numbers from $i$ to $j$, i.e. ${\N_i^j \doteq \{i, i+1, \dots, j-1, j\}}$. 
Given two vectors $x$ and $y$, $x \leq (\geq) \; y$ denotes componentwise inequalities.
The set of positive definite matrices of dimension $n$ is given by $\Sp{n}$, whereas $\Dp{n}$ is the set of \textit{diagonal} positive definite matrices of dimension $n$.
Given vectors $x_j$ defined for a (finite) index set $j \in \cc{J} \subset \N$, we denote by a bold $\vv{x}$ their Cartesian product.
We denote a (non-finite) sequence of vectors $x_j$ indexed by $j \in \N$ by $\{x\}$.
Given a vector $x\inR{n}$, we denote its $i$-th component using a parenthesized subindex $x_{(i)}$.
The set of non-negative real numbers is denoted by $\R_+$.
A function $\alpha: \R_+ \rightarrow \R$ is a $\cc{K}_\infty$-class function if it is continuous, strictly increasing, unbounded above and $\alpha(0) = 0$.
Given a symmetric matrix $A$, we denote by $\lambda_\text{max}(A)$ and $\lambda_\text{min}(A)$ its largest and smallest eigenvalues, respectively.
Given two vectors $x\inR{n}$ and $y\inR{n}$, their standard inner product is denoted by $\sp{x}{y} \doteq \Sum{i=1}{n} x_{(i)} y_{(i)}$.
For a vector $x\inR{n}$ and a matrix $A\in\Sp{n}$, $\|x\| \doteq \sqrt{\sp{x}{x}}$ and $\|x\|_A$ denotes the weighted Euclidean norm $\|x\|_A \doteq \sqrt{\sp{x}{A x}}$.

\section{Problem formulation} \label{sec:Problem:Formulation}

We consider a controllable linear time-invariant system described by the following discrete state space model,
\begin{equation} \label{eq:Model}
    x_{k+1} = A x_k + B u_k,
\end{equation}
where $x_k \inR{n}$ and $u_k \inR{m}$ are the state and control input at sample time $k$, respectively. Additionally, we consider that the system is subject to the box constraint,
\begin{equation} \label{eq:Constraints}
    (x_k, u_k) \in \cc{Z} \doteq \{(x, u) {\in} \R^n {\times} \R^m : \zLB \leq C x {+} D u \leq \zUB \},
\end{equation}
where we assume that the upper and lower bounds $\zLB \inR{n_z}$, $\zUB \inR{n_z}$, satisfy $z_m < z_M$.

We are interested in controllers capable of steering the system to the given reference $(x_r, u_r)$ while satisfying the system constraints \eqref{eq:Constraints}. This will only be possible if the reference is an \textit{admissible} steady state of the system, as defined in the following definition. Otherwise, we wish the system to be steered to the ``closest'' admissible steady state to the reference, for some given criterion of closeness.

\begin{definition} \label{def:Admissible}
    An ordered pair $(x_a, u_a) \in \R^n \times \R^m$ is said to be an admissible steady state of system (\ref{eq:Model}) subject to (\ref{eq:Constraints}) if
    \bRlist
        \item $x_a = A x_a + B u_a$, i.e. it is a steady state of system (\ref{eq:Model}).
        \item $(x_a, u_a) \in \ri{\cc{Z}}$.
    \eRlist
\end{definition}

\begin{remark} \label{rem:Admissible}
    We note that the imposition of condition \textit{(ii)} in the previous definition, instead of $(x_a, u_a) \in \cc{Z}$, is necessary to avoid the possible controllability loss when the constraints are active at the equilibrium point \cite{Limon_A_2008}. In a practical setting, this restriction is typically substituted by defining vectors ${\hzLB \doteq \zLB + \epsilon}$ and ${\hzUB \doteq \zUB - \epsilon}$, where $\epsilon \inR{n_z}$ is some arbitrarily small positive vector such that $\hzLB < \hzUB$, and instead imposing $(x_a, u_a) \in \hat{\cc{Z}} \subset \ri{\cc{Z}}$, where
\begin{equation*}
\hat{\cc{Z}} \doteq \set{(x, u) \in \R^n \times \R^m}{\hzLB \leq C x + D u \leq \hzUB},
\end{equation*}
This way, we avoid working with the open set $\ri{\cc{Z}}$.
\end{remark}

\section{Model Predictive Control for tracking} \label{sec:MPCT}

This section recalls the MPC \textit{for tracking} (MPCT) formulation \cite{Ferramosca_A_2009}, which is the basis of the controller we propose in Section \ref{sec:HMPC}.

In MPCT, an artificial reference is included as an additional decision variable of the optimization problem. This inclusion provides several benefits, such as \textit{(i)} a significant increase of the domain of attraction w.r.t. standard MPC formulations, \textit{(ii)} recursive feasibility, even in the event of reference changes, and \textit{(iii)} the use of a terminal set which is valid for any reference. In what follows, we present a specific MPCT formulation which uses a singleton as the terminal set and which does not require a terminal cost.

For a given prediction horizon $N$, the MPCT control law for a given state $x$ and reference $(x_r, u_r)$ is derived from the solution of the following convex optimization problem labeled by $\lT(x; x_r, u_r)$,
\begin{subequations} \label{eq:MPCT} 
\begin{align}  
    \lT(x; x_r, u_r) \doteq &\min\limits_{\vv{x}, \vv{u}, x_a, u_a} \; J(\vv{x}, \vv{u}, x_a, u_a; x_r, u_r) \\ 
    s.t.& \; x_{j+1} = A x_j + B u_j, \; j\in\N_0^{N-1} \\ 
        & \; \zLB \leq C x_j + D u_j \leq \zUB, \; j\in\N_0^{N-1} \\
        & \; x_0 = x \\
        & \; x_N = x_a \label{eq:MPCT:Terminal} \\
        & \; x_a = A x_a + B u_a \label{eq:MPCT:Steady:State}\\
        & \; \hzLB \leq C x_a + D u_a \leq \hzUB, \label{eq:MPCT:z_a}
\end{align}
\end{subequations}
where $\vv{x} = ( x_0, \dots, x_{N-1}, x_{N} )$ are the predicted states, $\vv{u} = ( u_0, \dots, u_{N-1} )$ the control inputs, and $(x_a, u_a)$ is the artificial reference.
The cost function $J(\vv{x}, \vv{u}, x_a, u_a; x_r, u_r)$ is composed of two terms: the summation of stage costs
\begin{equation*} \label{eq:State:Cost:MPCT}
    \ell_T(\vv{x}, \vv{u}, x_a, u_a) = \sum\limits_{j=0}^{N-1} \| x_j - x_a \|_{Q}^{2} + \sum\limits_{j=0}^{N-1} \| u_j - u_a \|_{R}^{2}, 
\end{equation*}
which penalizes the distance between the predicted states $x_j$ and inputs $u_j$ with the artificial reference by means of the cost function matrices $Q \in \Sp{n}$ and $R \in \Sp{m}$; and the offset cost 
\begin{equation} \label{eq:Terminal:Cost:MPCT}
    V_T(x_a, u_a; x_r, u_r) = \| x_a - x_r \|^{2}_{T_a} + \| u_a - u_r \|^{2}_{S_a}, 
\end{equation}
which penalizes the distance between the artificial reference $(x_a, u_a)$ and $(x_r, u_r)$ by means of the cost function matrices ${T_a \in \Sp{n}}$ and ${S_a \in \Sp{m}}$.
Note that equations \eqref{eq:MPCT:Steady:State} and \eqref{eq:MPCT:z_a}, where $\hzLB$ and $\hzUB$ are obtained as in Remark \ref{rem:Admissible}, guarantee that $(x_a, u_a)$ is an admissible reference (see Def. \ref{def:Admissible}).

The MPCT formulation guarantees that the closed-loop system asymptotically converges to an admissible steady state of the system so long as the problem is initially feasible, regardless of whether or not the reference $(x_r, u_r)$ is admissible \cite{Limon_A_2008, Ferramosca_A_2009}. In fact, if the reference is not admissible, the system will converge to the admissible steady state that minimizes the offset cost function \eqref{eq:Terminal:Cost:MPCT}.

\section{Harmonic based MPC for tracking} \label{sec:HMPC}

This section presents the main contribution of the paper: a harmonic based MPC formulation for tracking. The idea behind this formulation is to substitute the artificial reference of MPCT with the artificial harmonic reference sequences $\{x_h\}$, $\{u_h\}$, whose values at each discrete time instant $j\in\N$ are given by,
\begin{subequations} \label{eq:harmonic:signals}
\begin{align}
    \xhj &= \xe + \xs \sin(w ({j}{-}{N})) + \xc \cos(w ({j}{-}{N})), \label{eq:x_h}\\
    \uhj &= \ue + \us \sin(w ({j}{-}{N})) + \uc \cos(w ({j}{-}{N})), \label{eq:u_h}
\end{align}
\end{subequations}
where $w > 0$ is the base frequency. The harmonic sequences $\{x_h\}$ and $\{u_h\}$ are parameterized by decision variables $\xe$, $\xs$, ${\xc \inR{n}}$ and $u_e$, $\us$, ${\uc \inR{m}}$. To simplify the text, we use the following notation,
\begin{align*}
    &\xH \doteq (\xe, \xs, \xc) \in \R^n \times \R^n \times \R^n, \\
    &\uH \doteq (\ue, \us, \uc) \in \R^m \times \R^m \times \R^m,
\end{align*}
\vspace{-2em}
\begin{align} \label{eq:def:z}
\bmat{ccc} \ze & \zs & \zc \emat \doteq \bmat{cc} C & D \emat \bmat{ccc} \xe & \xs & \xc \\ \ue & \us & \uc \emat.
\end{align}

For a given prediction horizon $N$ and base frequency $w$, the HMPC control law for a given state $x$ and reference $(x_r, u_r)$ is derived from the following second order cone programming problem labeled by $\lH(x; x_r, u_r)$,
\begin{subequations} \label{eq:HMPC} 
\begin{align}
    \moveEq{-12} \lH(&x; x_r, u_r) \doteq \min\limits_{\vv{x},\vv{u}, \xH, \uH} \; \costFunH(\vv{x}, \vv{u}, \xH, \uH; x_r, u_r) \\ 
    \moveEq{-12}  &s.t.\; x_{j+1} = A x_j + B u_j, \; j\in\N_0^{N-1} \label{eq:HMPC:dynamics}\\
                  & \zLB \leq C x_j + D u_j \leq \zUB, \; j\in\N_0^{N-1} \label{ineq:HMPC:z:first}\\
  & x_0 = x \label{eq:HMPC:cond:inic}\\
  & x_N = \xe + \xc \label{eq:HMPC:xN}\\
  & \xe = A \xe + B \ue \label{eq:HMPC:xe}\\
  & \xs \cos(w) - \xc \sin(w) = A \xs + B \us \label{eq:HMPC:xs}\\
  & \xs \sin(w) + \xc \cos(w) = A \xc + B \uc \label{eq:HMPC:xc}\\
  & \sqrt{ \zsj^2 + \zcj^2 } \leq \zej - \hzLBj, \; i\in\N_1^{n_z} \label{ineq:HMPC:z:minus}\\
  & \sqrt{ \zsj^2 + \zcj^2 } \leq \hzUBj - \zej, \; i\in\N_1^{n_z}, \label{ineq:HMPC:z:plus}
\end{align}
\end{subequations}
where $\vv{x} = \{ x_0, \dots, x_{N-1} \}$, $\vv{u} = \{ u_0, \dots, u_{N-1} \}$, and the cost function $\costFunH(\vv{x}, \vv{u}, \xH, \uH; x_r, u_r)$ is composed of two terms: the summation of stage costs
\begin{equation*} \label{eq::HMPC:Stage:Cost}
    \stageCostH(\vv{x}, \vv{u}, \xH, \uH) = \Sum{j=0}{N-1} \| x_j - \xhj \|_Q^2 + \| u_j - \uhj \|_R^2,
\end{equation*}
where $Q\in\Sp{n}$ and $R\in\Sp{m}$; and the offset cost
\begin{align} \label{eq:HMPC:Offset:Cost}
    \moveEq{-10}  \offsetCostH(\xH, \uH &; x_r, u_r) = \| \xe - x_r \|_{T_e}^2 + \| \ue - u_r \|_{S_e}^2 \nonumber \\
                                      &+ \| \xs \|_{T_h}^2 + \| \xc \|_{T_h}^2 + \| \us \|_{S_h}^2 + \| \uc \|_{S_h}^2,
\end{align}
where $T_e\in\Sp{n}$, $T_h \in\Dp{n}$, $S_e\in\Sp{m}$, and $S_h \in\Dp{m}$.

Note that the artificial harmonic reference \eqref{eq:harmonic:signals} is taking the role of the artificial reference $(x_a, u_a)$ of the MPCT controller. Therefore, it must satisfy the system dynamics and the tightened constraints (see Remark \ref{rem:Admissible}) at all future time instants. The reason for choosing a harmonic signal \eqref{eq:harmonic:signals}, as opposed to some other periodic signal, is that these two conditions can be imposed by the addition of the small amount of constraints \eqref{eq:HMPC:xe} to \eqref{ineq:HMPC:z:plus}, the number of which does not depend on the value of the prediction horizon nor on the period of the artificial harmonic reference. In particular, as shown in Property \ref{prop:periodic:dynamics} in Appendix \ref{app:properties}, the satisfaction of the system dynamics is imposed with the inclusion of constraints \eqref{eq:HMPC:xe} to \eqref{eq:HMPC:xc}. Additionally, as shown by Property \ref{prop:bounds} in Appendix \ref{app:properties}, the satisfaction of the tightened constraints is imposed with the inclusion of constraints \eqref{ineq:HMPC:z:minus} and \eqref{ineq:HMPC:z:plus}.

Constraint \eqref{eq:HMPC:xN} imposes that the predicted state $x_N$ reaches the harmonic artificial reference (note that $x_{hN} = x_e +x_c$). Then, since we are imposing the satisfaction of the system dynamics and tightened constraints with \eqref{eq:HMPC:xe}-\eqref{ineq:HMPC:z:plus}, we can remain in an admissible state trajectory \eqref{eq:x_h} by applying the (admissible) control actions given by \eqref{eq:u_h}.

We denote the optimal value of optimization problem \eqref{eq:HMPC} for a state $x$ and a given reference $(x_r, u_r)$ by $\lH^*(x; x_r, u_r) {=} \costFunH(\vv{x}^*, \vv{u}^*, \xH^*, \uH^*; x_r, u_r)$, where $\vv{x}^*$, $\vv{u}^*$, $\xH^*$, $\uH^*$ are the arguments that minimize \eqref{eq:HMPC}. Furthermore, for every $j \in \N$, we denote by
\begin{subequations} \label{eq:optimal:harmonic:signals}
\begin{align}
    \xhj^* &= \xe^* + \xs^* \sin(w ({j}{-}{N})) + \xc^* \cos(w ({j}{-}{N})), \label{eq:optimal:harmonic:signals:x} \\
    \uhj^* &= \ue^* + \us^* \sin(w ({j}{-}{N})) + \uc^* \cos(w ({j}{-}{N})), \label{eq:optimal:harmonic:signals:u}
\end{align}
\end{subequations}
the harmonic signals parameterized by $(\xH^*, \uH^*)$.
At each sample time $k$, the HMPC control law is given by $u_k = u_0^*$, obtained from the solution of $\lH(x_k; x_r, u_r)$.

Note that the constraints \eqref{eq:HMPC:dynamics}-\eqref{ineq:HMPC:z:plus} do not depend on the reference. Therefore, the feasibility region of the HMPC controller, i.e. the set of states $x$ for which $\lH(x; x_r, u_r)$ is feasible, is independent of the reference. As such, feasibility is never lost in the event of reference changes.

Theorem \ref{theo:HMPC:Stability} states the asymptotic stability of the HMPC controller to the \textit{optimal artificial harmonic reference}, which is defined and characterized below in Definition \ref{def:optimal:artificial:reference:HMPC}. In order to prove it, we first prove the recursive feasibility of the HMPC controller, which is stated in Theorem \ref{theo:Recursive:Feasibility}. This theorem was originally stated in \cite[Theorem 1]{Krupa_CDC_19} without its proof, which we include in Appendix \ref{app:proof:recursive:feasibility} of this manuscript.

\begin{definition}[Optimal artificial harmonic reference] \label{def:optimal:artificial:reference:HMPC}
    Given a reference $(x_r, u_r)$, we define the \textit{optimal artificial harmonic reference} of the HMPC controller as the harmonic sequences $\{\xh^\circ\}$, $\{\uh^\circ\}$ (see \eqref{eq:harmonic:signals}) parameterized by the unique solution $(\xH^\circ, \uH^\circ)$ of the strongly convex optimization problem
\begin{align} \label{eq:OP:optimal:harmonic:refrefence:HMPC}
    (\xH^\circ, \uH^\circ) = &\arg\min\limits_{\xH, \uH} \offsetCostH(\xH, \uH; x_r, u_r) \\
                               &s.t.\; \eqref{eq:HMPC:xe} \text{-} \eqref{ineq:HMPC:z:plus}. \nonumber
\end{align}
Additionally, we denote by $\offsetCostH^\circ(x_r, u_r) \doteq \offsetCostH(\xH^\circ, \uH^\circ; x_r, u_r)$ the optimal value of problem \eqref{eq:OP:optimal:harmonic:refrefence:HMPC}.
\end{definition}

The following lemma states that the \textit{optimal artificial harmonic reference} is in fact an admissible steady state of system \eqref{eq:Model} subject to \eqref{eq:Constraints}, i.e. $\xhj^\circ = \xeo$, $\uhj^\circ = \ueo$ and $(\xhj^\circ, \uhj^\circ) \in \hat{\cc{Z}}$, $\forall j\in\N$.

\begin{lemma} \label{lemma:optimal:artificial:reference:HMPC} 
    Consider optimization problem \eqref{eq:OP:optimal:harmonic:refrefence:HMPC}. Then, for any $(x_r, u_r)$, its optimal solution is the admissible steady state $(\xeo, \ueo) \in \hat{\cc{Z}}$ that minimizes $\| \xeo - x_r \|^2_{T_e} + \| \ueo - u_r \|^2_{S_e}$. That is, $\xHo = (\xeo, 0, 0)$ and $\uHo = (\ueo, 0, 0)$.
\end{lemma}

\begin{proof} 
    We prove the lemma by contradiction. Assume that $\hat{\vv{x}}_H^\circ = (\xeo, \xs^\circ, \xc^\circ)$, $\hat{\vv{u}}_H^\circ = (\ueo, \us^\circ, \uc^\circ)$, is the optimal solution of \eqref{eq:OP:optimal:harmonic:refrefence:HMPC} and that at least some (if not all) of $\xs^\circ$, $\xc^\circ$, $\us^\circ$, $\uc^\circ \neq 0$. First, we show that $\vv{x}_H^\circ = (\xeo, 0, 0)$, $\vv{u}_H^\circ = (\ueo, 0, 0)$ satisfy \eqref{eq:HMPC:xe}-\eqref{ineq:HMPC:z:plus}. Constraints \eqref{eq:HMPC:xs} and \eqref{eq:HMPC:xc} are trivially satisfied and \eqref{eq:HMPC:xe} is satisfied since $(\hat{\vv{x}}_H^\circ, \hat{\vv{u}}_H^\circ)$ is assumed to be the solution of \eqref{eq:OP:optimal:harmonic:refrefence:HMPC}. Moreover, since
    $$0 \leq \sqrt{ (\zsj^\circ)^2 + (\zcj^\circ)^2},\; \forall i\in\N_1^{n_z},$$
we have that \eqref{ineq:HMPC:z:minus} and \eqref{ineq:HMPC:z:plus} are also satisfied for $(\vv{x}_H^\circ, \vv{u}_H^\circ)$.
Finally, it is clear from the initial assumption and \eqref{eq:HMPC:Offset:Cost} that
$$\offsetCostH(\vv{x}_H^\circ, \vv{u}_H^\circ; x_r, u_r) < \offsetCostH(\hat{\vv{x}}_H^\circ, \hat{\vv{u}}_H^\circ; x_r, u_r),$$
contradicting the optimality of $(\hat{\vv{x}}_H^\circ, \hat{\vv{u}}_H^\circ)$. The fact that $(\xeo, \ueo) \in \hat{\cc{Z}}$ follows from the satisfaction of \eqref{ineq:HMPC:z:minus}-\eqref{ineq:HMPC:z:plus}.
\end{proof}

The following theorem states the recursive feasibility of the HMPC controller. That is, suppose that a state $x$ belongs to the feasibility region of the HMPC controller. Then, for any feasible solution $\vv{x}$, $\vv{u}$, $\vv{x}_H$ and $\vv{u}_H$ of $\lH(x; x_r, u_r)$ we have that the successor state $A x + B u_0$ also belongs to the feasibility region of the HMPC controller. The first claim of the theorem states that feasible solutions of the HMPC controller provide constraint satisfaction for all future predicted states. The second claim states the recursive feasibility property of the HMPC controller. Its proof, which we include in Appendix \ref{app:proof:recursive:feasibility}, follows the standard approach, in which it is shown that a feasible solution of \eqref{eq:HMPC} can be be obtained for the successor state from the feasible solution of the previous time instant.

\begin{theorem}[Recursive feasibility of the HMPC controller] \label{theo:Recursive:Feasibility} 
    Suppose that $x$ belongs to the feasibility region of the HMPC controller. Suppose also that ${\bar{\vv{x}} = \{\bx_0, \ldots, \bx_{N-1}\}}$, ${\bar{\vv{u}} = \{\bu_0, \ldots, \bu_{N-1}\}}$, $\xe$, $\xs$, $\xc$, $\ue$, $\us$, $\uc$, constitute a feasible solution to the constraints (\ref{eq:HMPC:dynamics}) to (\ref{ineq:HMPC:z:plus}). Then,
\bRlist
    \item The control input sequence $\{u\}$ defined as
        \begin{equation}
            u_j = \bsis{l} \bu_j, \; j\in\N_0^{N-1} \\ 
            \uhj, \; j \geq N, \\ \esis \label{eq:def:u:j}
        \end{equation}
        where $\uhj$ is given by \eqref{eq:u_h}, and the trajectory $\{x\}$ defined as $x_0=x$,
        $$ x_{j+1} = A x_j + B u_j, \;j\geq 0, $$
        satisfies
        \begin{equation} \label{ineq:z:totales}
            \zLB \leq C x_j + D u_j \leq \zUB , \forall j\geq 0.
        \end{equation}
    \item The successor state $Ax+B\bar{u}_0$ also belongs to the feasibility region of the HMPC controller.
\eRlist
\end{theorem}
\begin{proof} \renewcommand{\qedsymbol}{}
    See Appendix \ref{app:proof:recursive:feasibility}.
\end{proof}

Theorem \ref{theo:HMPC:Stability} states the asymptotic stability of the HMPC controller to the \textit{optimal artificial harmonic reference} (Def.~\ref{def:optimal:artificial:reference:HMPC}). Its proof, which is inspired by the proof of the asymptotic stability for the MPCT controller \cite[Theorem 1]{Limon_TAC_2018}, relies on the following well known Lyapunov asymptotic stability theorem \cite[Appendix B.3]{Rawlings_MPC_2017}. However, in our case, we directly derive a Lyapunov function that satisfies the asymptotic stability conditions. An assumption made by Theorem \ref{theo:HMPC:Stability} is that $N$ must be greater or equal to the controllability index of the system, which we define below in Definition \ref{def:controllability:index}. Note that the controllability index of a controllable system is always lower or equal to the dimension of its state space.

\begin{theorem}[Lyapunov asymptotic stability] \label{theo:Lyapunov:Stability} 
    Consider an autonomous system $z_{k+1} = f(z_k)$ with states $z_k \inR{n}$ and where the function $f: \R^n \rightarrow  \R^n$ is continuous and satisfies $f(0) = 0$. Let $\Gamma$ be a positive invariant set and $\Omega \subseteq \Gamma$ be a compact set, both including the origin as an interior point. If there exists a function $W:\R^n \rightarrow \R_+$ and suitable $\cc{K}_\infty$-class functions $\alpha_1(\cdot)$ and $ \alpha_2(\cdot)$ such that,
\bRlist
    \item $W(z_k) \geq \alpha_1(\| z_k \|), \; \forall z_k \in \Gamma,$
    \item $W(z_k) \leq \alpha_2(\| z_k \|), \; \forall z_k \in \Omega,$
    \item $W(z_{k+1}) < W(z_k), \forall z_k \in \Gamma \setminus \{0\},$ \\ and $W(z_{k+1}) = W(z_k)$ if $z_k = 0$,
\eRlist
then $W(\cdot)$ is a Lyapunov function for $z_{k+1} = f(z_k)$ in $\Gamma$ and the origin is asymptotically stable for all initial states in $\Gamma$.
\end{theorem}

\begin{definition}[Controllability index] \label{def:controllability:index}
    Consider a controllable system \eqref{eq:Model}. Its controllability index is the smallest integer ${j > 0}$ for which matrix $[B, A B, A^2 B \dots A^{j-1} B]$ has rank $n$.
\end{definition}

\begin{theorem}[Asymptotic stability of the HMPC controller] \label{theo:HMPC:Stability} 
    Consider a controllable system \eqref{eq:Model} subject to \eqref{eq:Constraints} and assume that $N$ is greater or equal to its controllability index. Then, for any reference $(x_r, u_r)$ and initial state $x$ belonging to the feasibility region of the HMPC controller $\lH(x; x_r, u_r)$, the system controlled by the control law derived from the solution of \eqref{eq:HMPC} is stable, fulfills the system constraints at all future time instants, and asymptotically converges to the optimal artificial harmonic reference $(\xeo, \ueo)$ given by Lemma~\ref{lemma:optimal:artificial:reference:HMPC}.
\end{theorem}
\begin{proof} \renewcommand{\qedsymbol}{}
See Appendix \ref{app:proof:stability:HMPC}.
\end{proof}

Note that, as stated in Theorem \ref{theo:HMPC:Stability}, the HMPC controller provides asymptotic convergence to an admissible steady state regardless of whether the reference is itself an admissible steady state or not. In fact, as is also the case with the MPCT controller \cite{Limon_TAC_2018}, it is clear from Lemma \ref{lemma:optimal:artificial:reference:HMPC} that the HMPC controller will converge to $(x_r, u_r)$ if it is admissible, and that it will converge to the admissible steady state $(x, u)$ that minimizes the distance $\| x - x_r \|_{T_e}^2 + \| u - u_r \|_{S_e}^2$ otherwise.

\section{Closed-loop comparison between the HMPC and MPCT controllers} \label{sec:HMPC:performance}

This section presents results of controlling a ball and plate system with the HMPC and MPCT controllers. The inclusion of the MPCT controller is done to highlight the fact that for certain systems, and especially for low values of the prediction horizon, its closed-loop performance can suffer due to the fact that the predicted state $x_N$ must reach an admissible steady state of the system (see constraint \eqref{eq:MPCT:Terminal}). The results with the HMPC controller, and our subsequent discussion in Section \ref{sec:simulation}, suggest that HMPC may provide a significant improvement over MPCT in this regard.

\subsection{Ball and plate system} \label{sec:ball:and:plate}

The ball and plate system consists of a plate that pivots around its center point such that its slope can be manipulated by changing the angle of its two perpendicular axes. The objective is to control the position of a solid ball that rests on the plate. We assume that the ball is always in contact with the plate and that it does not slip when moving. The non-linear equations of the system are \cite{Wang_ISA_2014},
\begin{subequations} \label{eq:BaP:nonlinear}
\begin{align}
    \ddot{\zb}_1 &= \frac{m}{m + I_b/r^2} \left( \zb_1 \dot{\theta}_1^2 + \zb_2 \dot{\theta}_1 \dot{\theta}_2 + g \sin{\theta_1} \right) \\
    \ddot{\zb}_2 &= \frac{m}{m + I_b/r^2} \left( \zb_2 \dot{\theta}_2^2 + \zb_1 \dot{\theta}_1 \dot{\theta}_2 + g \sin{\theta_2} \right),
\end{align}
\end{subequations}
where $m$, $r$ and $I_b$ are the mass, radius and mass moment of inertia of a solid ball, respectively; $\zb_1$ and $\zb_2$ are the position of the ball on the two axes of the plate relative to its center point; $\dot{\zb}_1$, $\dot{\zb}_2$, $\ddot{\zb}_1$ and $\ddot{\zb}_2$ their corresponding velocities and accelerations; $\theta_1$ and $\theta_2$ are the angle of the plate on each of its axes; and $\dot{\theta}_1$ and $\dot{\theta}_2$ their corresponding angular velocities.

The state of the system is given by
\begin{equation*}
    x = (\zb_1, \dot{\zb}_1, \theta_1, \dot{\theta}_1, \zb_2, \dot{\zb}_2, \theta_2, \dot{\theta}_2),
\end{equation*}
and the control input $u = (\ddot{\theta}_1, \ddot{\theta}_2)$ is the angle acceleration of the plate in each one of its axes. We consider the following constraints on the velocity, angles and control inputs,
\begin{align*}
    |\dot{\zb}_i| \leq 0.5\,\text{m/s}^2, \;
    |\theta_i| \leq \frac{\pi}{4}\,\text{rad}, \;
    |\ddot{\theta}_i| \leq 0.4\,\text{rad/s}^2, \; i \in\N_1^2.
\end{align*}

A linear time-invariant discrete-time model \eqref{eq:Model} of the system is obtained by linearizing its non-linear equations taking the origin as the operating point and discretizing with a sample time of $0.2$s.
We use this linear model as the prediction model of the MPC controllers as well as the model used to simulate the system. We do so to illustrate the properties of the HMPC controller under nominal conditions, since Theorems \ref{theo:Recursive:Feasibility} and \ref{theo:HMPC:Stability} consider this premise.
We take $m = 0.05\,$Kg, $r = 0.01\,$m, $g = 9.81\,$m/s$^2$ and $I_b = (2/5) m r^2 = 2\cdot 10^{-6}$Kg$\cdot$m$^2$.

\subsection{Performance comparison between HMPC and MPCT} \label{sec:simulation}

\begin{figure*}[t]
    \centering
    \begin{subfigure}[ht]{0.48\textwidth}
        \includegraphics[width=\linewidth]{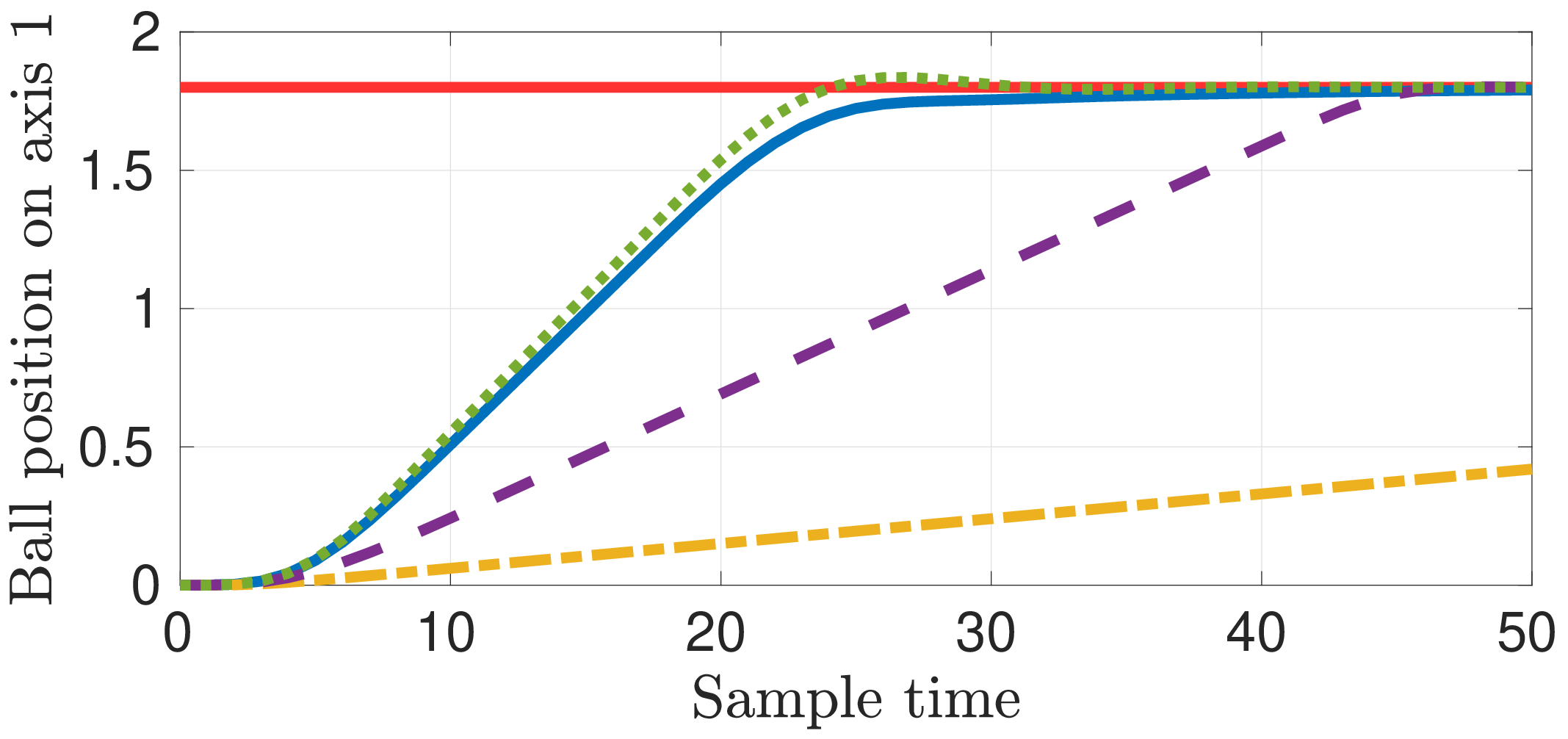}
        \caption{Position of ball on axis 1.}
        \label{fig:comparison:position}
    \end{subfigure}%
    \hfill
    \begin{subfigure}[ht]{0.48\textwidth}
        \includegraphics[width=\linewidth]{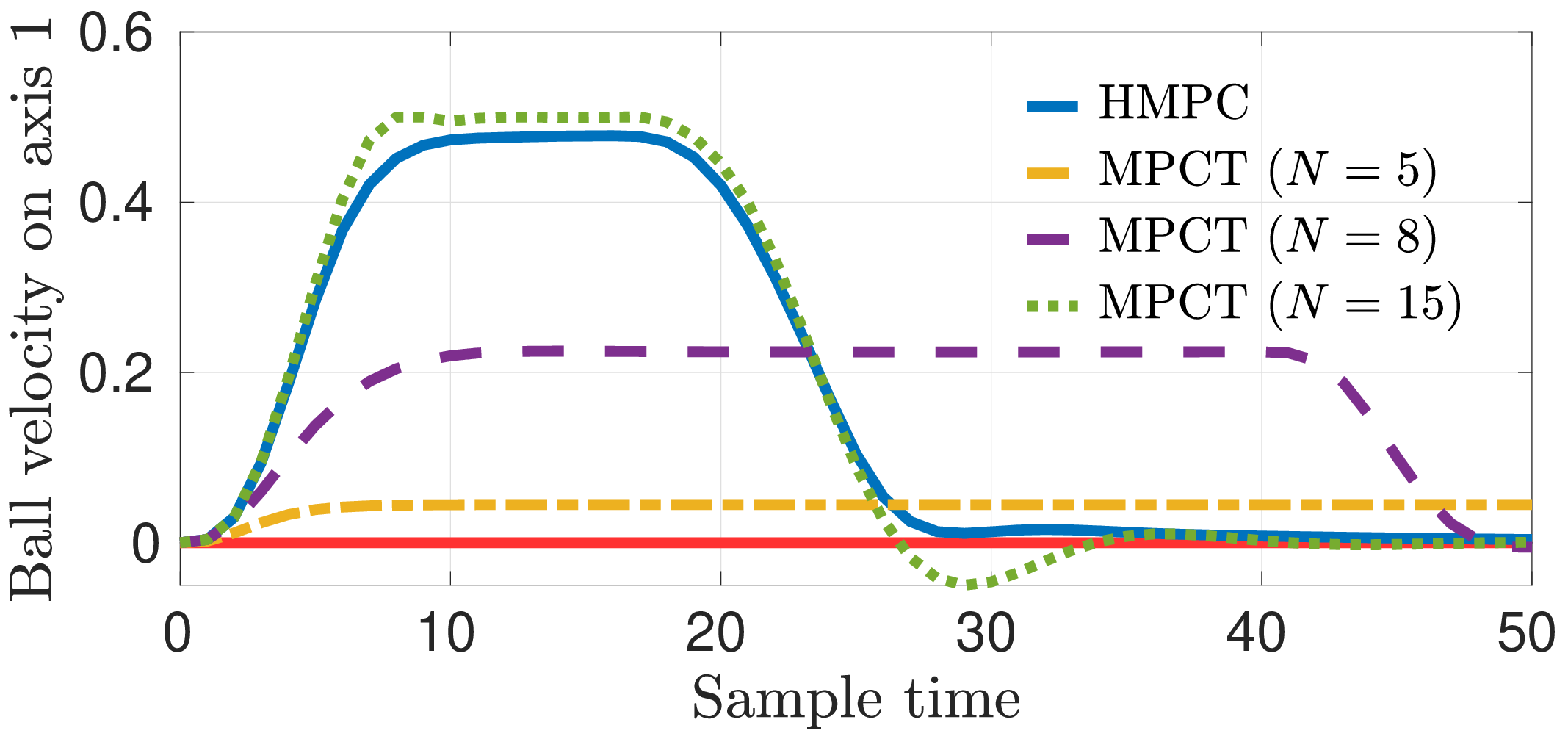}
        \caption{Velocity of ball on axis 1.}
        \label{fig:comparison:velocity}
    \end{subfigure}%

    \begin{subfigure}[ht]{0.48\textwidth}
        \includegraphics[width=\linewidth]{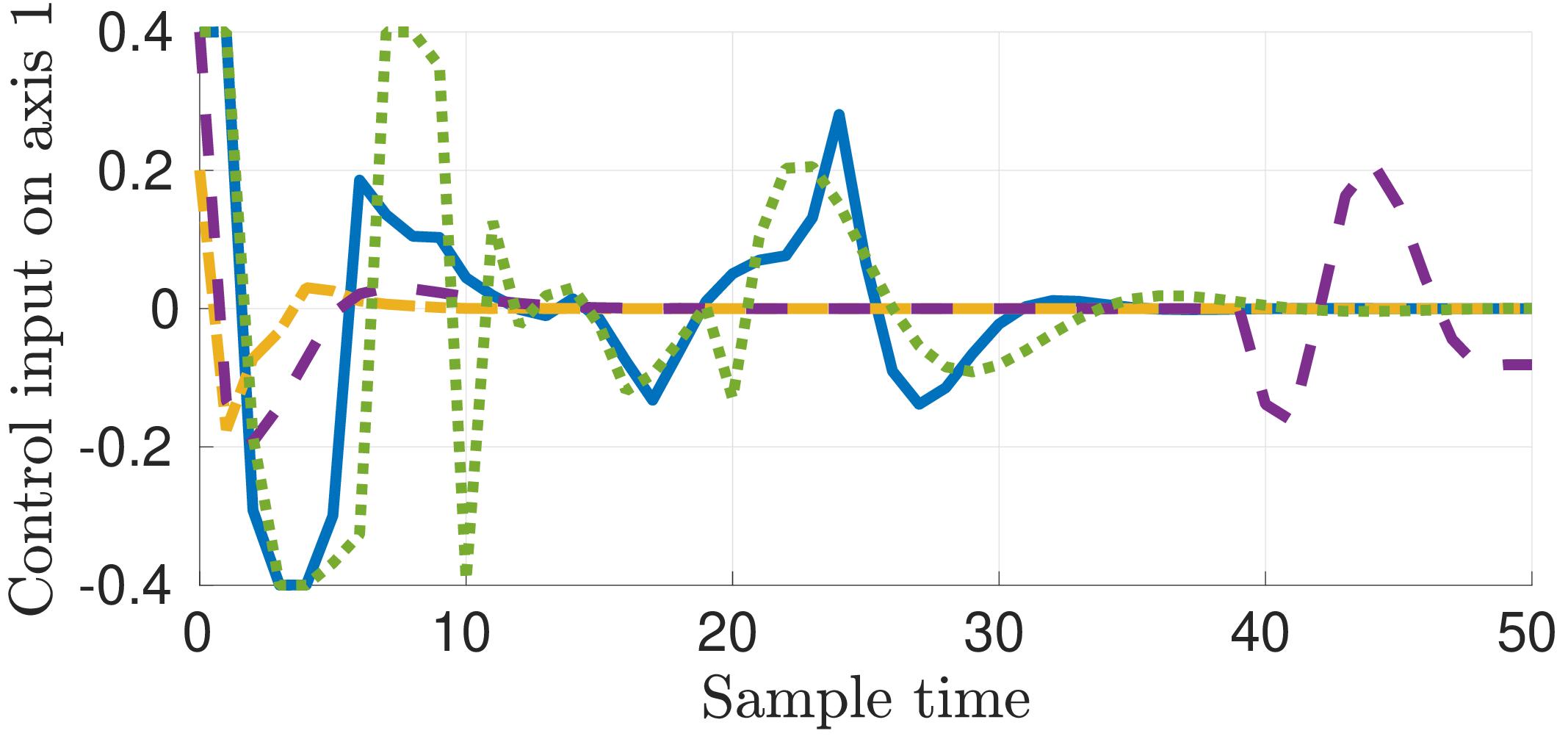}
        \caption{Control input on axis 1.}
        \label{fig:comparison:input}
    \end{subfigure}%
    \hfill
    \begin{subfigure}[ht]{0.48\textwidth}
        \includegraphics[width=\linewidth]{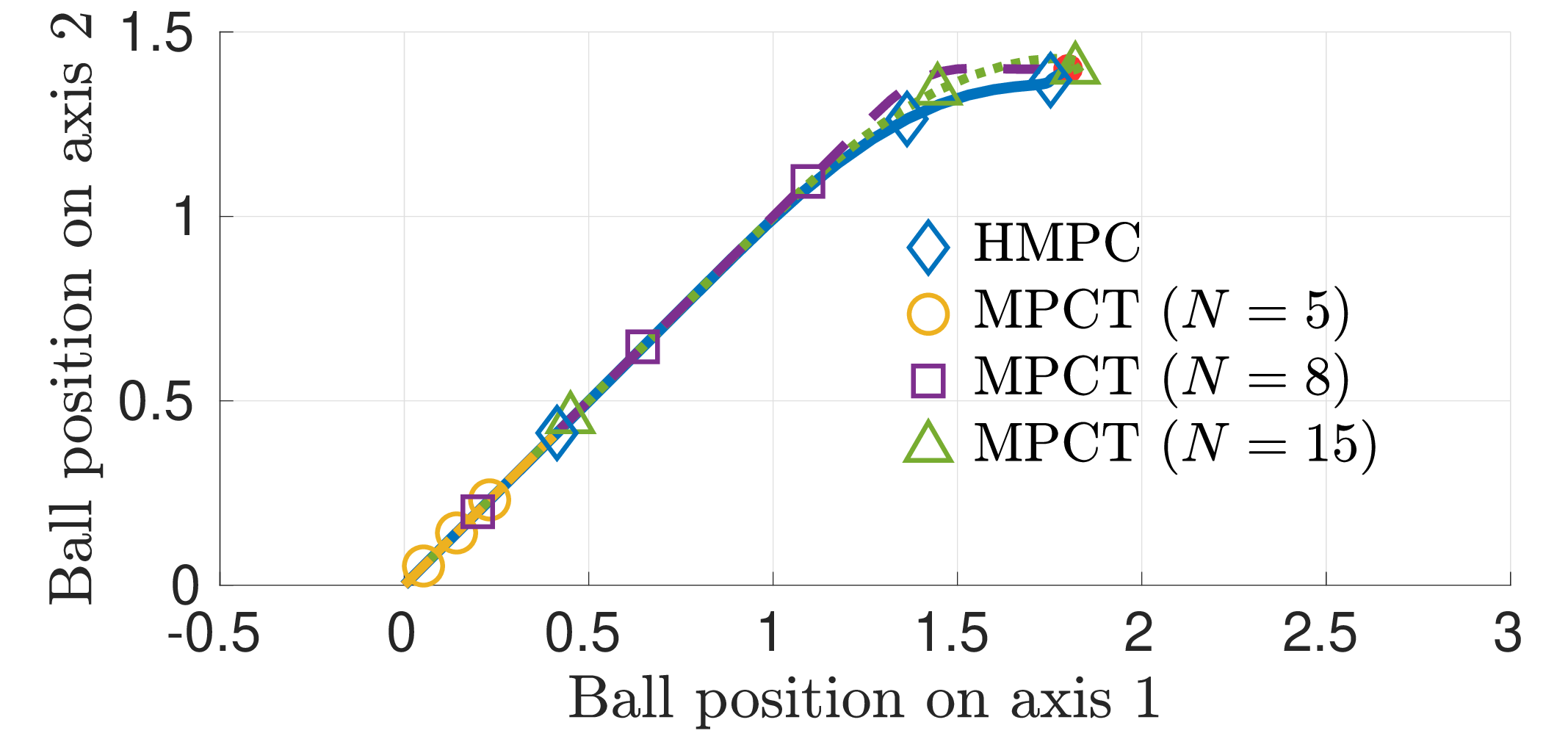}
        \caption{Position of ball on the plate.}
        \label{fig:comparison:plate}
    \end{subfigure}%
    \caption{Closed-loop comparison between HMPC and MPCT.}
    \label{fig:comparison}
\end{figure*}

\begin{figure}[t]
    \centering
    \includegraphics[width=\columnwidth]{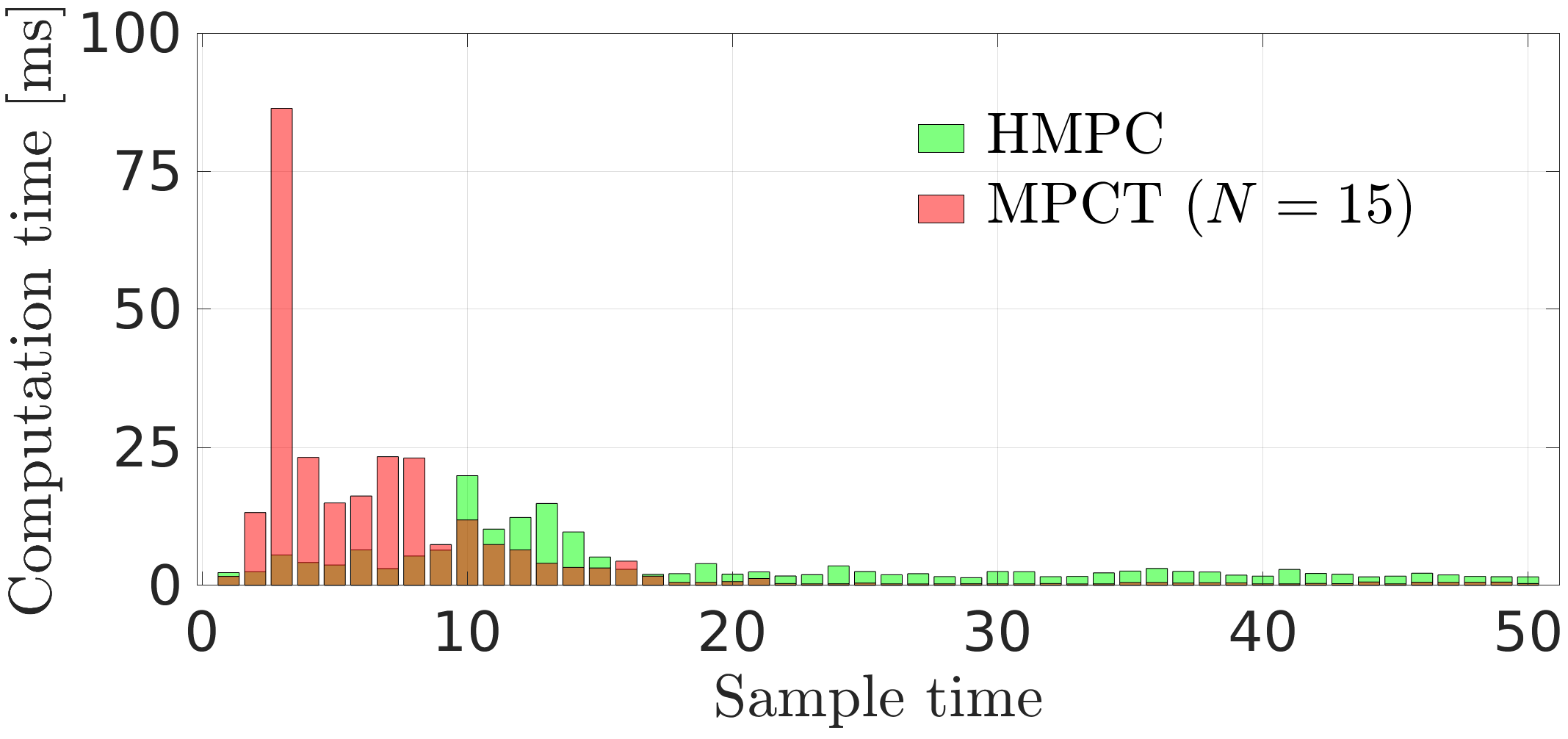}
    \caption{Computation times of the COSMO and OSQP solvers.}
    \label{fig:computation:times}
\end{figure}

\begin{figure*}[t]
    \centering
    \begin{subfigure}[ht]{0.48\textwidth}
        \includegraphics[width=\linewidth]{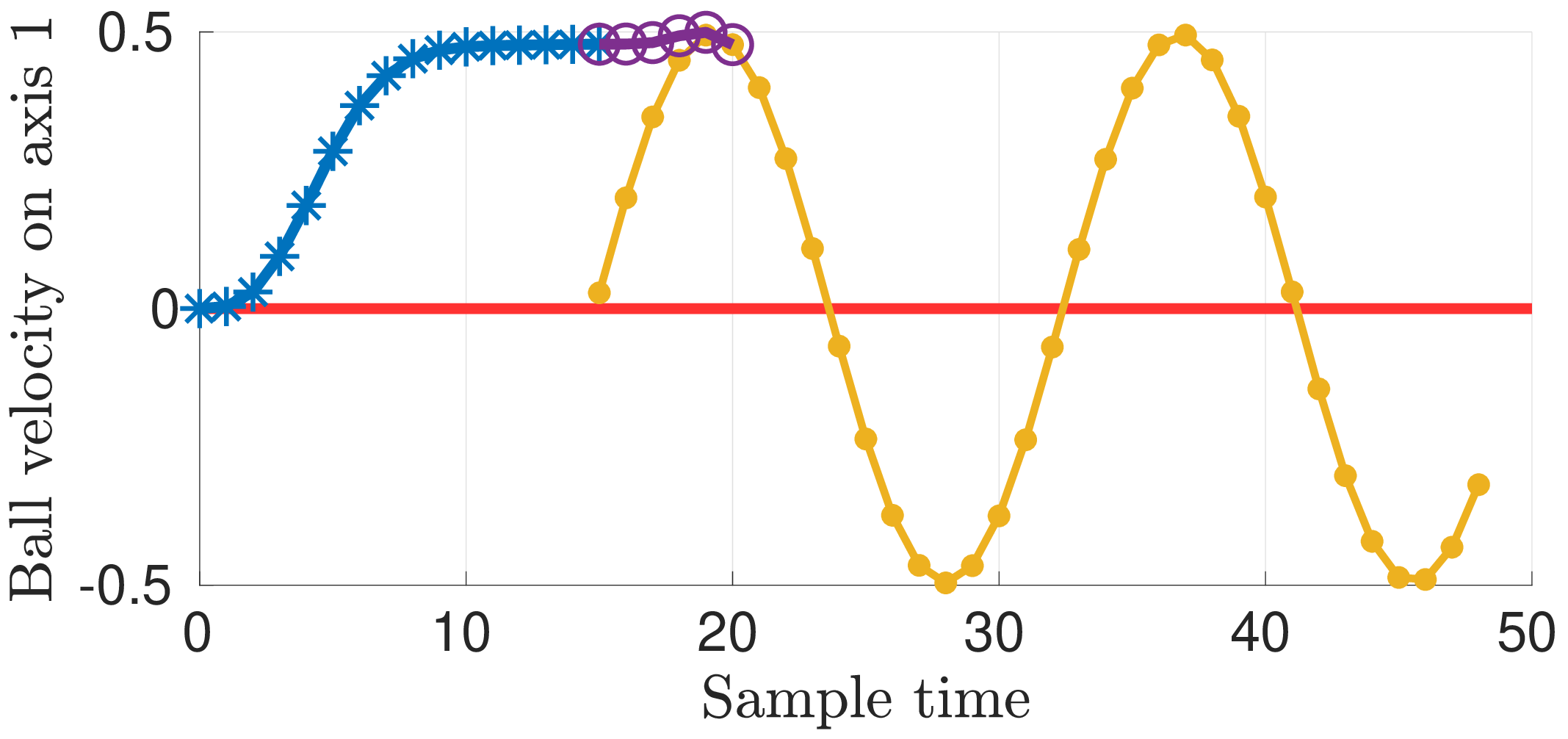}
        \caption{HMPC: Velocity of ball on axis 1.}
        \label{fig:iter:HMPC:vel}
    \end{subfigure}%
    \hfill
    \begin{subfigure}[ht]{0.48\textwidth}
        \includegraphics[width=\linewidth]{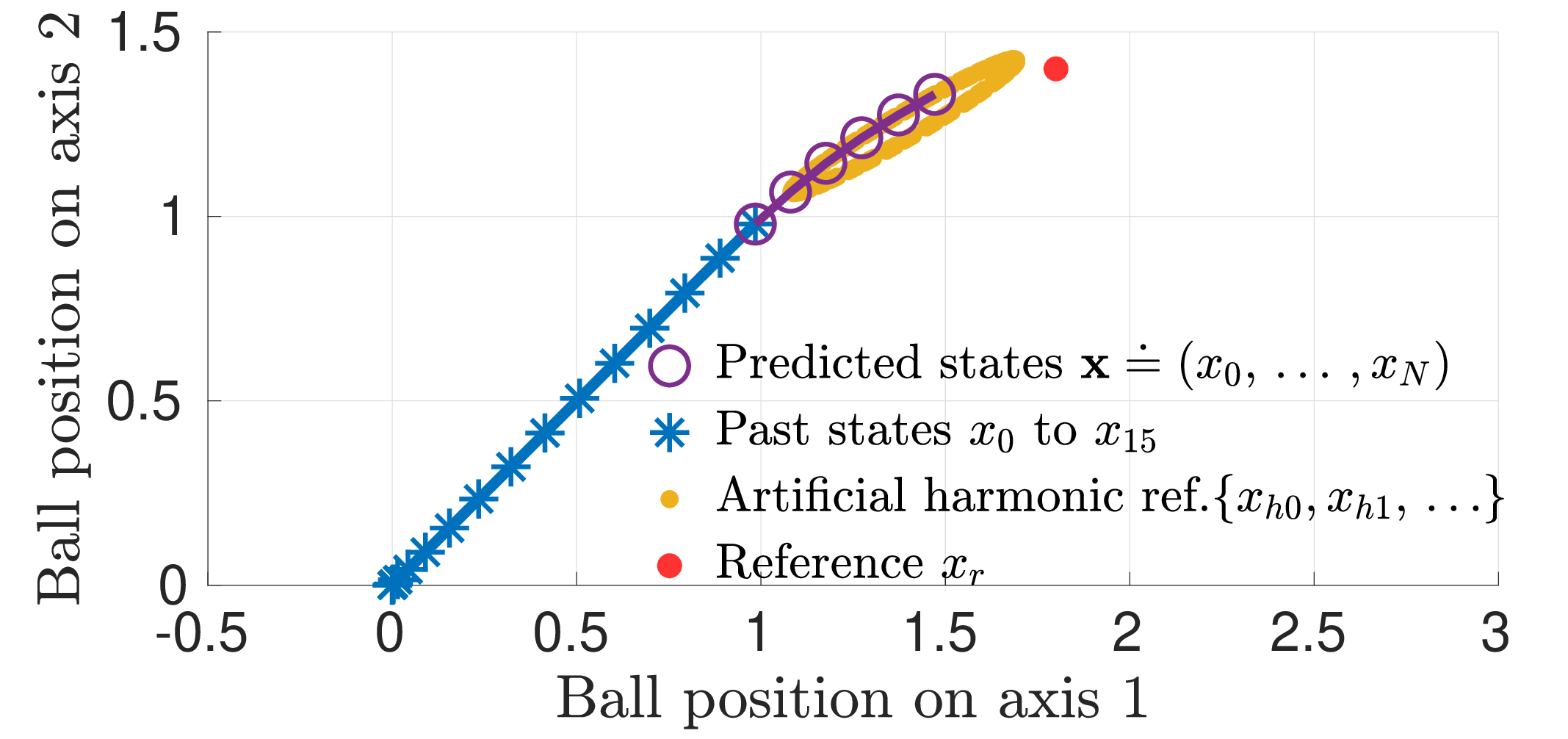}
        \caption{HMPC: Position of ball on plate.}
        \label{fig:iter:HMPC:pos}
    \end{subfigure}%

    \begin{subfigure}[ht]{0.48\textwidth}
        \includegraphics[width=\linewidth]{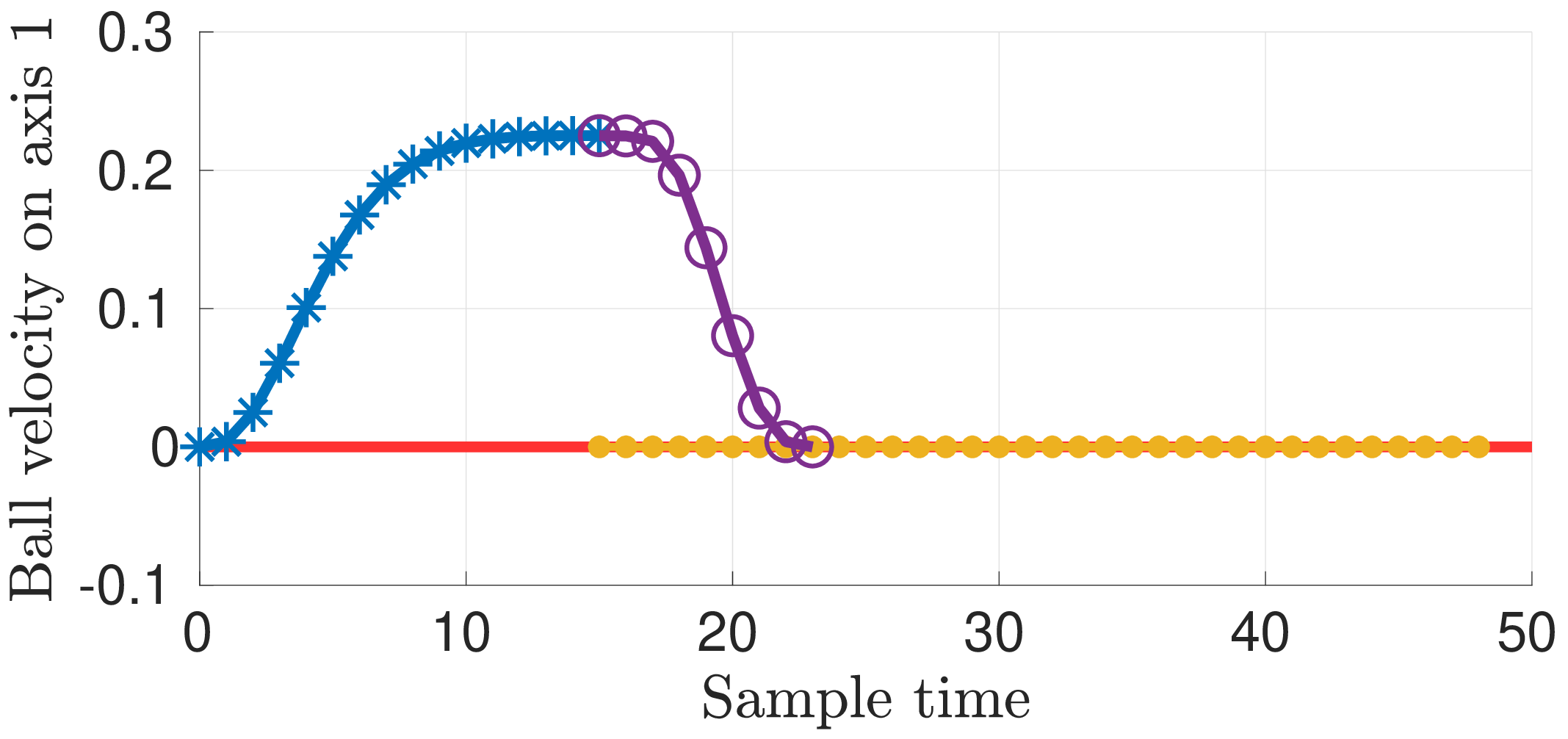}
        \caption{MPCT: Velocity of ball on axis 1.}
        \label{fig:iter:MPCT:vel}
    \end{subfigure}%
    \hfill
    \begin{subfigure}[ht]{0.48\textwidth}
        \includegraphics[width=\linewidth]{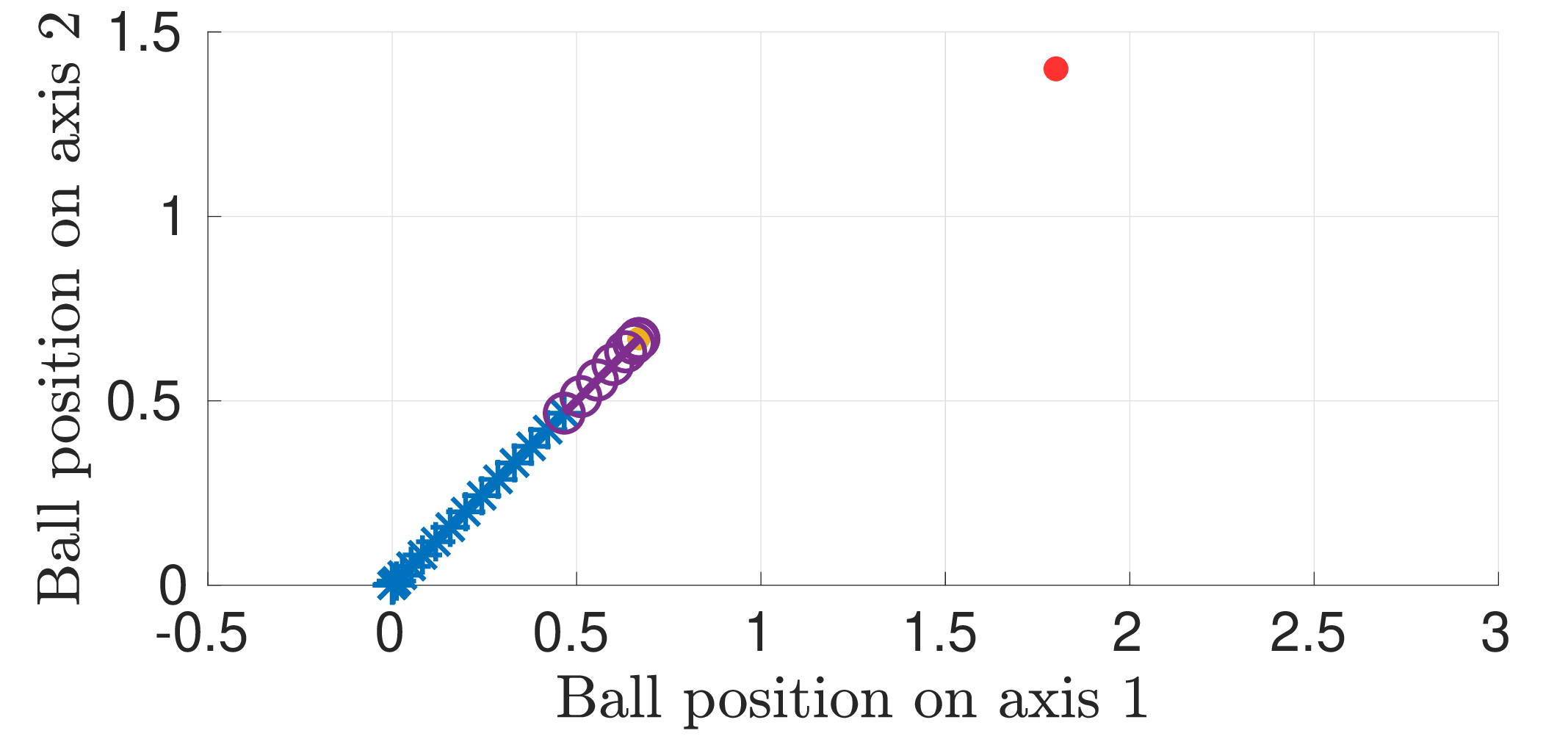}
        \caption{MPCT: Position of ball on plate.}
        \label{fig:iter:MPCT:pos}
    \end{subfigure}%
    \caption{Snapshot of HMPC and MPCT at iteration 15.}
    \label{fig:iter}
\end{figure*}

We perform a closed-loop simulation of the ball and plate system with the MPCT and HMPC controllers. The system is initialized at the origin and the objective is to steer it to the position $\zb_1 = 1.8$, $\zb_2 = 1.4$, i.e. 
$$x_r = (1.8, 0, 0, 0, 1.4, 0, 0, 0), \; u_r = (0, 0).$$

The HMPC controller is solved using version v0.7.5 of the COSMO solver \cite{Garstka_COSMO_ECC_2019}, while the MPCT controller is solved using version 0.6.0 of the OSQP solver \cite{Stellato_OSQP}. These two solvers employ the same operator splitting approach, based on the alternating direction method of multipliers \cite{Boyd_FTML_2011}. In fact, their algorithms are very similar, with OSQP being particularized to QP problems. The settings of both solvers are set to their default values with the exception of the tolerances \texttt{eps\_abs}, \texttt{eps\_rel}, \texttt{eps\_prim\_inf} and \texttt{eps\_dual\_inf}, which are set to $10^{-4}$.

The parameters of the controllers, which where manually tuned to provide an adequate closed-loop performance, are described in Table \ref{tab:MPC:parameters}. We compare the HMPC controller with $N=5$ to three MPCT controllers with prediction horizons $N= 5, 8, 15$. The prediction horizon $N=15$ was chosen by finding the lowest value for which the MPCT performed well. The performance is measured as
$$\Phi \doteq \Sum{k=1}{N_\text{iter}} \| x_k - x_r \|^2_Q + \| u_k - u_r \|^2_R,$$
where $x_k$, $u_k$ are the states and control actions throughout the simulation and $N_\text{iter} = 50$ is the number of sample times. Table \ref{tab:performance:index:comparison} shows the performance index for each one of the controllers.

{\renewcommand{\arraystretch}{1.2}%
    \begin{table}[t]
    \centering
\begin{threeparttable}
    \caption{Parameters of the controllers}
    \label{tab:MPC:parameters}
    \begin{tabular}{llll}
        \toprule
        Parameter        & \multicolumn{3}{l}{Value} \\
        \midrule
        $Q$ & \multicolumn{3}{l}{\textit{diag}$( 10, 0.05, 0.05, 0.05, 10, 0.05, 0.05, 0.05)$} \\
        $T_e$ & \multicolumn{3}{l}{\textit{diag}$( 600, 50, 50, 50, 600, 50, 50, 50)$} \\
        \midrule
        Parameter        & Value                      & Parameter    & Value                    \\
        \midrule
        $R$              & \textit{diag}$(0.5, 0.5)$  & $S_e$        & \textit{diag}$(0.3, 0.3)$ \\
        $T_h$            & $T_e$                      & $S_h$        & $0.5 S_e$                 \\
        $T_a$            & $T_e$                      & $S_a$        & $S_e$                 \\
        $N$              & $5$ (HMPC), $8$ and $15$ & $\epsilon$   & $( 10^{-4}, 10^{-4}, 10^{-4})$\\
        $w$              & $0.3254$ & & \\
        \bottomrule
    \end{tabular}
    \begin{tablenotes}[flushleft] \footnotesize
    \item \textit{diag}$(\cdot)$ denotes a diagonal matrix with the indicated elements.
    \end{tablenotes}
\end{threeparttable}
\end{table}}

{\renewcommand{\arraystretch}{1.1}%
    \begin{table}[t]
    \centering
    \caption{Performance comparison between controllers}
    \label{tab:performance:index:comparison}
    \begin{tabular}{ccccc}
        \toprule
        Controller & \multicolumn{3}{c}{MPCT} & HMPC \\
        \cmidrule(lr){2-4}\cmidrule(lr){5-5}
        Prediction horizon $(N)$ & 5 & 8 & 15 & 5 \\
        \midrule
        Performance $(\Phi)$ & $2014.03$ & $844.16$ & $488.88$  & $511.09$ \\
        \bottomrule
    \end{tabular}
\end{table}}

Figure \ref{fig:comparison} shows the closed-loop simulation results for each controller. Figures \ref{fig:comparison:position} and \ref{fig:comparison:velocity} show the position and velocity of the ball on axis 1, i.e. $\zb_1$ and $\dot{\zb}_1$, respectively. Figure \ref{fig:comparison:input} shows the control input on axis 1, i.e. $\ddot{\theta}_1$. Finally, Figure \ref{fig:comparison:plate} shows the trajectory of the ball on the plate. The markers indicate the position of the ball at sample times $10$, $20$ and $30$ for each one of the controllers. The computation times of the HMPC controller and the MPCT controller with the prediction horizon $N=15$ are shown in Figure \ref{fig:computation:times}. As can be seen, the COSMO solver applied to the HMPC problem provides computation times that are reasonable when compared to the results of the OSQP solver applied to the MPCT problem; in spite of the fact that the OSQP solver is particularized to QP problems, whereas the COSMO solver is not particularized to the second order cone programming problem \eqref{eq:HMPC}. Our intent with this figure is to show that, even though \eqref{eq:HMPC} is a more complex problem than \eqref{eq:MPCT} due to the inclusion of second order cone constraints, it can still be solved in reasonable times using state of the art solvers. We note that COSMO runs on the \textit{Julia} programming language, while OSQP is programmed in \textit{C} and executed using its Matlab interface.

Notice that the velocities obtained with the MPCT controllers with small prediction horizons are far away from its upper bound of $0.5$. The HMPC controller, on the other hand, reached much higher velocities even though its prediction horizon is also small. This results in a much faster convergence of the HMPC controller, as can be seen in Figures \ref{fig:comparison:position} and \ref{fig:comparison:plate}. If the prediction horizon of the MPCT controller is sufficiently large (e.g. $N=15$), then this issue no longer persists.

To understand why this happens, let us compare the solution of the HMPC controller with the MPCT controller with ${N=8}$. Figure \ref{fig:iter} shows a snapshot of sample time $15$ of the same simulation shown in Figure \ref{fig:comparison}. Lines marked with an asterisk are the past states from iteration $k = 0$ to the \textit{current} state at iteration $k = 15$, those marked with circumferences are the predicted states $\vv{x}$ for $j\in\N_0^N$, and those marked with dots are the artificial reference. The position of the markers line up with the value of the signals at each sample time, e.g. each asterisk marks the value of the state at each sample time $k\in\N_0^{15}$. Figures \ref{fig:iter:HMPC:vel} and \ref{fig:iter:MPCT:vel} show the velocity $\dot{\zb}_1$ of the ball on axis 1 for the HMPC and MPCT controllers, respectively. Figures \ref{fig:iter:HMPC:pos} and \ref{fig:iter:MPCT:pos} show the position of the ball on the plate.

The reason why the velocity does not exceed ${\approx}{0.2}$ with the MPCT controller can be seen in Figure \ref{fig:iter:MPCT:vel}. The predicted states of the MPCT controller must reach a steady state at $j = N$ (see constraint \eqref{eq:MPCT:Terminal}). In our example this translates into the velocity having to be able to reach $0$ within a prediction window of length $N=8$. This is the reason that is limiting the velocity of the ball. A velocity of $0.5$ is not attainable with an MPCT controller with a prediction horizon of $N=8$ because there are no admissible control input sequences $\vv{u}$ capable of steering the velocity from $0.5$ to $0$ in $8$ sample times. This issue does not occur with the HMPC controller because it does not have to reach a steady state at the end of the prediction horizon, as can be seen in Figure \ref{fig:iter:HMPC:vel}. Instead, it must reach an admissible ``steady state" harmonic reference, which can have a non-zero velocity.

It is clear from this discussion, and the results of the MPCT controller with $N=15$, that this issue will become less and less pronounced as the prediction horizon is increased. However, for low values of the prediction horizon, the HMPC controller can provide a significantly better performance than the MPCT controller, as shown in the example presented here.

Figure \ref{fig:simulation:nonlinear} illustrates the behaviour of the HMPC controller when simulating the system using the non-linear model \eqref{eq:BaP:nonlinear}. In this test, the system is started in the origin. Then, each $50$ sample times, the reference is changed to each one of the vertices of a regular pentagon inscribed in the unit circle (in clockwise order starting from the uppermost vertex). Additionally, Gaussian noise with a standard deviation of $0.01$ is added to the position of the ball given to the HMPC controller as the current system state \eqref{eq:HMPC:cond:inic}. This is done to simulate a small amount of \textit{measurement} noise. The other states are left undisturbed for the sake of simplicity.

\begin{figure}[t]
    \centering
    \includegraphics[width=\columnwidth]{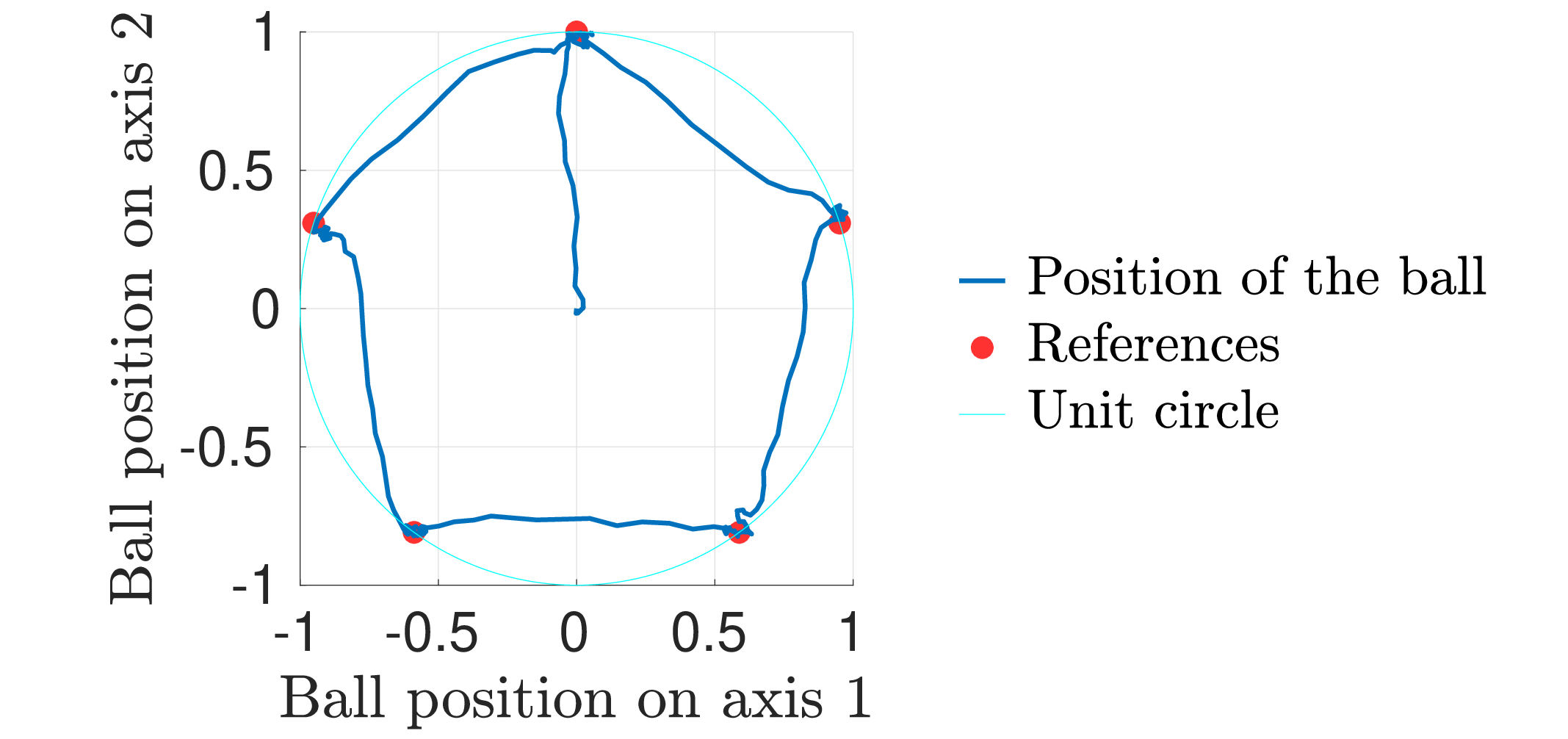}
    \caption{Closed-loop simulation of the non-linear model with the HMPC controller.}
    \label{fig:simulation:nonlinear}
\end{figure}

\begin{remark} \label{rem:MPCT:issue} 
    The performance advantages of a (suitably tuned) HMPC are especially noticeable if the system presents integrator states and/or slew rate constraints, as is the case of the example shown above. However, the issue that affects the performance of the MPCT controller is that the state cannot ``move far away" from the subspace of steady states of the system due to the presence of input constraints coupled with the low prediction horizon. As such, the performance advantage of the HMPC controller may still be present in a wider range of systems. Moreover, the HMPC controller can be viewed as an MPCT with an added degree of freedom.
    Therefore, it can always be tuned to behave like the MPCT controller, by taking, for instance, $w = 2 \pi$. Additionally, due to the added degree of freedom, the HMPC controller can provide an enlargement of the domain of attraction with respect to the MPCT controller, as illustrated in \cite{Krupa_CDC_19}.
\end{remark}

\section{Practical selection of parameter $w$} \label{sec:selection:w}

This section discusses the selection of parameter $w$ of the HMPC controller, providing a simple, intuitive approach for its selection. It is important to note that the stability and recursive feasibility of the controller are satisfied for any value of $w$. However, the performance of the controller can be improved by a proper selection of this parameter. 

There are two main considerations to be made. The first one is related to the phenomenon of aliasing and of the selection of the sampling time for continuous-time systems. This will provide an upper bound to $w$. The second one is related to the frequency response of linear systems, which will provide some insight into the selection of an initial, and well suited, value of $w$. Subsequent fine tuning may provide better results, but this initial value of $w$ should work well in practice and provide a good starting point.

\subsection{Upper bound of $w$} \label{sec:selection:w:upper:bound}

The signals \eqref{eq:harmonic:signals} parametrized by any feasible solution $(\xH, \uH)$ of \eqref{eq:HMPC} satisfy the discrete-time system dynamics, as discussed in Section \ref{sec:HMPC}. Therefore, all that remains is to select $w$ small enough such that signals \eqref{eq:harmonic:signals} describe a suitably sampled signal.

In order to prevent the aliasing phenomenon, $w$ must be chosen below the Nyquist frequency for anti-aliasing, i.e. ${w < \pi}$ \cite{Shannon_1949}. However, since the inputs are applied using a zero order holder, we would recommend taking
\begin{equation} \label{eq:selection:w:upper:bound}
    w \leq \frac{\pi}{2}.
\end{equation}

In any case, the stability and recursive feasibility of the controller will not be lost, since Theorems \ref{theo:Recursive:Feasibility} and \ref{theo:HMPC:Stability} do not make any assumptions on the value of $w$, but the benefits of using HMPC instead of MPCT may be lost if this bound is not respected. Indeed, for $w = 2\pi$, HMPC is identical to MPCT.

\subsection{Selection of a suitable $w$} \label{sec:selection:w:bode}

There are three additional considerations to be made for selecting an adequate $w$: \textit{(i)} high frequencies equate fast system responses, \textit{(ii)} high frequencies tend to have small input-to-state gains, and \textit{(iii)} the presence of state constraints.

At first glance, it would seem that selecting a high value of $w$ would lead to fast system responses. However, this need not be the case, since the gain of the system tends to diminish as the frequency of the input increases, i.e. if $w$ is selected in the \textit{high frequency} band of the system. If the gain is low, then $\xh$ is very similar to a constant signal of value $\xe$, which results in HMPC behaving very similarly to the MPCT. Therefore, $w$ should be selected taking into account the gain of the system for that frequency.

A tentative lower bound for $w$ is then the highest frequency of the \textit{low frequency} band of the system. However, a final consideration can be made with regard to the system constraints as follows: the presence of constraints can override the desire for frequencies with large system gains. For instance, take as an example a system with a static gain of $4$ with an input $u$ subject to $|u| \leq 1$ and a state $x$ subject to $|x| \leq 2$. Then, selecting a $w$ whose Bode gain is close to the static gain of the system is not desirable because the amplitude of $\uh$ will be limited by the constraints on $\xh$. Therefore, we can select a higher frequency. In this case, a proper selection might be to chose $w$ as the frequency whose Bode gain is $2$.

\begin{figure}[t]
    \centering
    \includegraphics[width=\columnwidth]{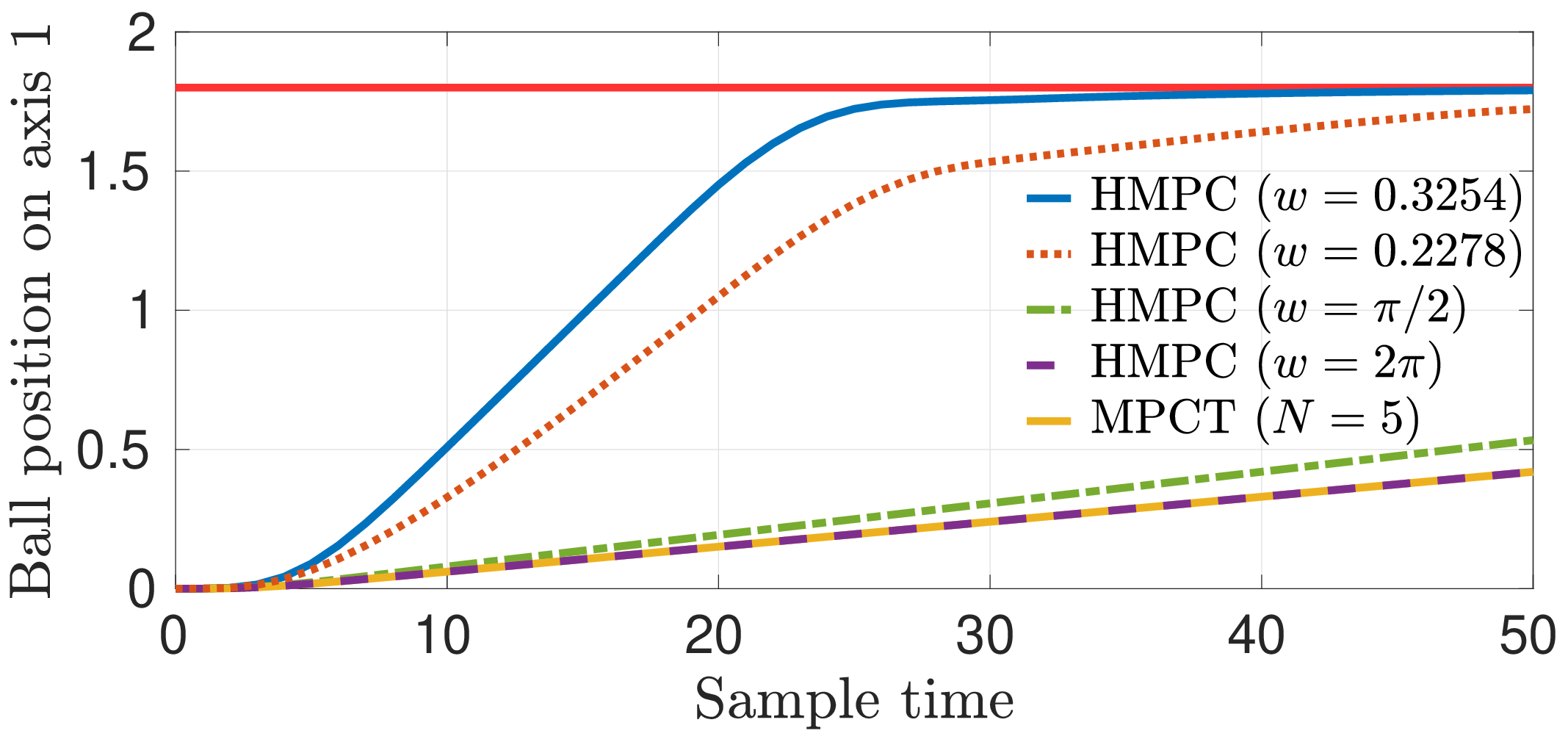}
    \caption{Closed-loop simulation for different realizations of $w$.}
    \label{fig:selection:w}
\end{figure}

\begin{remark} \label{rem:selection:w:MIMO}
    If the system has multiple states/inputs, then the above considerations should be made extrapolating the idea to the frequency response of MIMO systems. One approach in this case is to focus on the slow dynamics (states) of the system, which are the most restrictive, in that they may require higher prediction horizons in order to be able to reach steady states. Additionally, it is also useful to identify if the system has any integrator states and to take into account their constraints as described in the above discussion.
\end{remark}

\begin{example} \label{example:w_selection}
\normalfont
As an example, take the case study of Section \ref{sec:HMPC:performance}, which is a MIMO system. Figure \ref{fig:selection:w} shows the closed-loop simulation of HMPC controllers using the parameters from Table \ref{tab:MPC:parameters} but different realizations of $w$. For the results shown in Section \ref{sec:simulation}, $w = 0.3254$ was selected as the cutting frequency of the Bode plot from $u$ to $\dot{z}$. It was chosen this way because of the constraints $|u| \leq 0.4$ and $|\dot{z}| \leq 0.5$. As shown in the figure, choosing a lower $w$, such as $w = 0.7 \cdot 0.3254 = 0.2278$ would have resulted in a higher gain, which would be pointless due to these constraints, and an overall slower convergence due to the slower frequency and the smaller amplitude of $\{\uh\}$. On the other hand, choosing a higher $w$, such as $w = \pi/2$, leads to a small frequency gain. This results in a harmonic reference signal $\{\xh\}$ that is very similar to a constant signal, leading to a poor performance. Finally, we show one of the possible undesirable effects of choosing a $w$ that does not satisfy \eqref{eq:selection:w:upper:bound}. In this case, selecting $w = 2 \pi$ makes the HMPC controller identical to the MPCT controller with the same prediction horizon. We should note that all the simulations shown in Figure \ref{fig:selection:w} eventually converge to the reference.
\end{example}

\section{Conclusions} \label{sec:conclusions}

This paper presents a novel MPC formulation for tracking piecewise constant references that can significantly outperform other MPC formulations in the case of small prediction horizons, as well as provide a larger domain of attraction. This is due to the fact that the terminal state does not need to reach a steady state of the system, but instead just needs to reach a periodic trajectory of the system given by a single harmonic signal. 
We find that the performance advantage is especially noticeable in systems with integrators and/or slew rate constraints, which are very typical in robotic applications.

Additionally, the controller does not require the computation of a terminal set, and its recursive feasibility (and asymptotic stability) is guaranteed even in the event of a reference change. These properties are welcome in any setting, but particularly so when dealing with embedded systems.

The computation times needed to solve the HMPC problem with the COSMO solver suggest that its online implementation in embedded systems might be attainable, especially if a specialized solver is developed.

\newpage

\begin{appendix}

\subsection{Collection of properties} \label{app:properties}

This section contains three properties from the appendix of \cite{Krupa_CDC_19} which are used in various of the proofs of this manuscript. They are included here for completeness.

\begin{property} \label{prop:simple:armonics} 
Let the elements $v_\ell \inR{n_v}$ of a sequence $\{v\}$ be given by
\begin{equation*}
    v_\ell = v_e + v_s \sin(w \ell) + v_c \cos(w \ell),\; \forall \ell\in \N,
\end{equation*}
where $w\in \R$ and $v_e, v_s, v_c \inR{n_v}$.
Then,
\begin{equation*}
v_{\ell+1} = v_e + v_s^+ \sin(w \ell) + v_c^+ \cos(w \ell), \; \forall \ell \in \N,
\end{equation*}
where
\begin{align*}
    v_s^+ &= v_s \cos(w) - v_c \sin(w), \\
    v_c^+ &= v_s \sin(w) + v_c \cos(w).
\end{align*}
Moreover,
\begin{equation*}
(v_{s (i)}^+)^2 + (v_{c (i)}^+)^2 = v_{s (i)}^2 + v_{c (i)}^2, \; i\in\N_1^{n_v}.
\end{equation*}
\end{property}

\begin{proof} 
The proof relies on the following well-known trigonometric identities
\begin{align*}
\sin(\alpha+\beta) &= \sin(\alpha)\cos(\beta)+ \cos(\alpha)\sin(\beta) \\
\cos(\alpha+\beta) &= \cos(\alpha)\cos(\beta)-\sin(\alpha)\sin(\beta).
\end{align*}
From these expressions we obtain
\begin{align*}
    \sin(w (\ell+1)) &= \sin(w) \cos(w \ell) + \cos(w) \sin(w \ell) \\
    \cos(w (\ell+1)) &= \cos(w) \cos(w \ell) - \sin(w) \sin(w \ell).
\end{align*}
Therefore,
\begin{align*}
    v_{\ell+1}&= v_e + v_s \sin(w(\ell+1)) + v_c \cos(w(\ell+1)) \\
             &= v_e + v_s \left[ \sin(w)\cos(w\ell) + \cos(w)\sin(w\ell) \right]\\
             &\quad + v_c \left[ \cos(w)\cos(w\ell) - \sin(w)\sin(w\ell) \right] \\
             &= v_e + \left[ v_s \cos(w) - v_c \sin(w) \right] \sin(w\ell) \\
             &\quad + \left[ v_s \sin(w) + v_c \cos(w) \right] \cos(w\ell) \\
             &= v_e + v_s^+ \sin(w\ell) + v_c^+ \cos(w\ell).
\end{align*}
This proves the first claim of the property. Denote now
\begin{equation*}
\rm{H}_w \doteq \bmat{cc} \cos(w) & - \sin(w) \\ \sin(w) & \cos(w) \emat.
\end{equation*}
With this notation,
\begin{equation*}
\bv v_{s (i)}^+ \\ v_{c (i)}^+ \ev = \rm{H}_w \bv v_{s (i)} \\ v_{c (i)} \ev, \; i\in\N_1^{n_v}.
\end{equation*}
From the identity $\sin^2(w)+\cos^2(w)=1$ we obtain
\begin{equation*}
\rm{H}_w\T \rm{H}_w = I_2.
\end{equation*}
We are now in a position to prove the last claim of the property.
\begin{align*}
    &(v_{s (i)}^+)^2+ (v_{c (i)}^+)^2 = \left\| \bv v_{s (i)}^+ \\ v_{c (i)}^+ \ev \right\|^2 \\
    & = \bv v_{s (i)} \\ v_{c (i)} \ev\T \rm{H}_w\T \rm{H}_w \bv v_{s (i)} \\ v_{c (i)} \ev \\
    & = \bv v_{s (i)} \\ v_{c (i)} \ev\T \bv v_{s (i)} \\ v_{c (i)} \ev = v_{s (i)}^2 + v_{c (i)}^2. \qedhere
\end{align*}
\end{proof}

\begin{property}\label{prop:periodic:dynamics} 
Given the system $x_{k+1}=Ax_k+Bu_k$, suppose that
\begin{align*}
    & u_{N+\ell} = \ue + \us \sin(w \ell) + \uc \cos(w \ell), \; \forall  \ell \geq 0 \\
    & x_N = \xe + \xc \\
    & \xe = A \xe + B \ue \\
    & \xs \cos(w) - \xc \sin(w) =  A \xs + B \us \\
    & \xs \sin(w) + \xc \cos(w) =  A \xc + B \uc.
\end{align*}
Then
\begin{equation*}
x_{N+\ell} = \xe + \xs \sin(w\ell) + \xc \cos(w\ell), \; \forall \ell \geq 0.
\end{equation*}
\end{property}

\begin{proof} 
Since $x_N = \xe + \xc$, the claim is trivially satisfied for $\ell=0$. Suppose now that
the claim is satisfied for $\ell \geq 0$, we will show that it is also satisfied for $\ell +1$. Indeed,
\begin{align*}
    x_{N+\ell+1} &= A x_{N+\ell} + Bu_{N+\ell} \\
                 &=  A \left[ \xe + \xs \sin(w\ell) + \xc \cos(w\ell) \right] \\
                 &\quad + B \left[ \ue + \us \sin(w\ell) + \uc \cos(w\ell) \right] \\
                 &=  A \xe + B \ue + (A \xs + B \us) \sin(w\ell) \\
                 &\quad + (A \xc + B \uc) \cos(w\ell) \\
                 &= \xe + \left[ \xs \cos(w) - \xc \sin(w) \right] \sin(w\ell) \\
                 &\quad + \left[ \xs \sin(w) + \xc \cos(w) \right] \cos(w\ell) \\
                 &\numeq{*} \xe + \xs \sin(w(\ell+1)) + \xc \cos(w(\ell+1)).
\end{align*}
We note that equality $(*)$ is due to Property \ref{prop:simple:armonics}.
\end{proof}

\begin{property} \label{prop:bounds} 
    Let the elements $v_\ell \inR{n_v}$ of a sequence $\{v\}$ be given by
$$v_\ell = v_e + v_s \sin(w\ell) + v_c \cos(w\ell), \; \forall \ell \in \N,$$
where $w\in\R$ and $v_e, v_s, v_c\inR{n_v}$.
Then, for every $\ell \in\N$ and $i\in\N_1^{n_v}$, we have,
\begin{subequations}
\begin{align}
    v_{\ell (i)} &\leq v_{e (i)} + \sqrt{ v_{s (i)}^2 + v_{c (i)}^2}, \label{eq:bounds:upper} \\
    v_{\ell (i)} &\geq v_{e (i)} - \sqrt{ v_{s (i)}^2 + v_{c (i)}^2}. \label{eq:bounds:lower}
\end{align}
\end{subequations}
\end{property}

\begin{proof} 
    We prove inequality \eqref{eq:bounds:upper}. The proof for \eqref{eq:bounds:lower} is similar.
\begin{align*}
    v_{\ell (i)} & = v_{e (i)} + v_{s (i)} \sin(w\ell) + v_{c (i)} \cos(w\ell) \\
                 & = v_{e (i)} + \bmat{cc} v_{s (i)} & v_{c (i)} \emat \bv \sin(w\ell) \\ \cos(w\ell) \ev\\
                 & \leq v_{e (i)} + \left\| \bv v_{s (i)} \\ v_{c (i)} \ev \right\| \; \left\| \bv \sin(w\ell) \\ \cos(w\ell) \ev\right\| \\
                 & = v_{e (i)} + \sqrt{v_{s (i)}^2 + v_{c (i)}^2}. \qedhere
\end{align*}
\end{proof}

\subsection{Proof of the recursive feasibility of HMPC} \label{app:proof:recursive:feasibility}

\begin{proof}[Proof of Theorem \ref{theo:Recursive:Feasibility}] 
We begin by proving the first claim. Since $u_j=\bar{u}_j$ for $j\in\N_0^{N-1}$ and $x_0 = x$, we obtain by a direct inspection of \eqref{eq:HMPC:cond:inic} and \eqref{eq:HMPC:dynamics} that
\begin{equation} \label{eq:x:bar:x}
    x_j =\bx_j , \; j\in\N_0^N.
\end{equation}
This implies
\begin{equation*}
    C x_j + D u_j = C \bx_j + D \bu_j, \; j\in\N_0^{N-1}.
\end{equation*}
Therefore, we have from inequality \eqref{ineq:HMPC:z:first} that
\begin{equation}\label{ineq:z:hasta:N:minus}
    \zLB \leq C x_j + D u_j \leq \zUB, \; j\in\N_0^{N-1}.
\end{equation}
We now prove that these inequalities also hold for $j \geq N$.
From \eqref{eq:def:u:j} we have that
\begin{equation*}
    u_{N+\ell} = \ue + \us \sin(w\ell) + \uc \cos(w\ell), \; \ell \geq 0.
\end{equation*}
From \eqref{eq:x:bar:x} and \eqref{eq:HMPC:xN} we also have that $x_N = \bx_N = \xe + \xc$.
Taking also into consideration equalities \eqref{eq:HMPC:xe} to \eqref{eq:HMPC:xc} we obtain
\begin{align*}
    & u_{N+\ell} = \ue + \us \sin(w\ell) + \uc \cos(w\ell), \; \ell \geq 0 \\
    & x_N = \xe + \xc \\
    & \xe =A \xe + B \ue  \\
    & \xs \cos(w) - \xc \sin(w) = A \xs + B \us \\
    & \xs \sin(w) + \xc \cos(w) = A \xc + B \uc.
\end{align*}
which along with Property \ref{prop:periodic:dynamics}, allows us to write
\begin{equation*}
    x_{N+\ell} = \xe + \xs \sin(w\ell) + \xc \cos(w\ell), \; \forall \ell \geq 0.
\end{equation*}
Therefore, we have that
\begin{align*}
    z_{N+\ell} &= Cx_{N+\ell} + Du_{N+\ell} \\
               & = C \xe + D \ue + ( C \xs + D \us )\sin(w\ell) \\
               &\quad + ( C \xc + D \uc )\cos(w\ell)\\
               & = \ze + \zs \sin(w\ell) + \zc \cos(w\ell),
\end{align*}
where the last equality is simply due to the definitions of $\ze$, $\zs$ and $\zc$ \eqref{eq:def:z}.
From this expression of $z_{N+\ell}$ and Property \ref{prop:bounds} we deduce that for every  $\ell \geq 0$ and $i\in\N_1^{n_z}$,
\begin{align*} \label{ineq:z:in:abstract}
    z_{N+\ell, (i)} &\leq \zej + \sqrt{ \zsj^2 + \zcj^2}, \\
    z_{N+\ell, (i)} &\geq \zej - \sqrt{ \zsj^2 + \zcj^2}.
\end{align*}
From this, alongside inequalities \eqref{ineq:HMPC:z:minus} and \eqref{ineq:HMPC:z:plus}, we obtain that
\begin{equation*}
    \hzLBj \leq z_{N+\ell,(i)} \leq \hzUBj, \; \forall i\in\N_1^{n_z}, \; \forall \ell\geq 0.
\end{equation*}
Since by construction ${\zLBj \leq \hzLBj}$ and ${\hzUBj \leq \zUBj}$ (see Remark \ref{rem:Admissible}), we have that,
\begin{equation}\label{ineq:z:bounds}
    \zLB \leq C x_{N+\ell} + D u_{N+\ell} \leq \zUB, \; \forall \ell\geq 0.
\end{equation}
Which along with (\ref{ineq:z:hasta:N:minus}), proves (\ref{ineq:z:totales}).

We now prove the second claim, i.e. $A x + B \bu_0$ belongs to the feasibility region of $\lH(A x + B \bu_0; x_r, u_r)$. To do so, we show that
\begin{subequations} \label{eq:Shift:all}
\begin{align}
    & \bu_j^+ \doteq \bu_{j+1}, \; j\in\N_0^{N-2} \label{eq:Shift:u:shift} \\
    & \bu_{N-1}^+ \doteq \ue + \uc \label{eq:Shift:uN} \\
    & \bx_0^+ \doteq A x +B \bu_0 \label{eq:Shift:x:inic} \\
    & \bx_{j+1}^+ \doteq A \bx_j^+ +B \bu_j^+, \;j\in\N_0^{N-1} \label{eq:Shift:x:shift} \\
    & \ue^+ \doteq \ue \label{eq:Shift:ue} \\
    & \us^+ \doteq \us \cos(w) - \uc \sin(w) \\
    & \uc^+ \doteq \us \sin(w) + \uc \cos(w) \\
    & \bmat{ccc} \xe^+ & \xs^+ & \xc^+ \emat \doteq \bmat{cc} A & B \emat \bmat{ccc} \xe & \xs & \xc \\ \ue & \us & \uc \emat \label{eq:Shift:x_h}
\end{align}
\end{subequations}
is a feasible solution for the initial condition $A x + B \bu_0$ by showing that \eqref{eq:Shift:all} satisfies constraints \eqref{eq:HMPC:dynamics} to \eqref{ineq:HMPC:z:plus}. That is, we prove in what follows that
\begin{subequations}
\begin{align}
    & \bx_{j+1}^+ = A \bx_j^+ + B\bu_j^+,\; j\in\N_0^{N-1} \label{eq:Feas:dynamics}\\
    & \zLB \leq C \bx_j^+ + D \bu_j^+ \leq \zUB,\; j\in\N_0^{N-1} \label{ineq:Feas:z:first}\\
    & \bx^+_0 = A x + B \bu_0 \label{eq:Feas:cond:inic}\\
    & \bx_N^+ = \xe^+ + \xc^+ \label{eq:Feas:xN}\\
    & \xe^+ = A \xe^+ + B \ue^+ \label{eq:Feas:xe} \\
    & \xs^+ \cos(w) - \xc^+ \sin(w) = A \xs^+ + B \us^+ \label{eq:Feas:xs}\\
    & \xs^+ \sin(w) + \xc^+ \cos(w) = A \xc^+ + B \uc^+ \label{eq:Feas:xc}\\ 
    & \sqrt{ (\zsj^+)^2 + (\zcj^+)^2 } \leq \zej^+ - \hzLBj,\; i\in\N_1^{n_z} \label{ineq:Feas:z:minus}\\
    & \sqrt{ (\zsj^+)^2 + (\zcj^+)^2 } \leq \hzUBj - \zej^+,\; i\in\N_1^{n_z}, \label{ineq:Feas:z:plus}
\end{align}
\end{subequations}
where variables $\ze^+$, $\zs^+$ and $\zc^+$ are given by
\begin{equation*}
\bmat{ccc} \ze^+ & \zs^+ & \zc^+ \emat \doteq \bmat{cc} C & D \emat \bmat{ccc} \xe^+ & \xs^+ & \xc^+ \\ \ue^+ & \us^+ & \uc^+ \emat.
\end{equation*}

Equalities \eqref{eq:Feas:dynamics} and \eqref{eq:Feas:cond:inic} are trivially satisfied by construction (see \eqref{eq:Shift:x:inic}-\eqref{eq:Shift:x:shift}). Since $\bx_{0}^+ = A x +B \bu_0 = \bx_1$, and $\bu_j^+ = \bu_{j+1}$, $j\in\N_0^{N-2}$ (see (\ref{eq:Shift:u:shift})), we have
\begin{equation} \label{eq:equivalence:step:bxj_buj}
    \bv \bx^+_j \\ \bu_j^+ \ev = \bv \bx_{j+1} \\ \bu_{j+1} \ev, \;j\in\N_0^{N-2}.
\end{equation}
Therefore, from \eqref{ineq:HMPC:z:first} we obtain
\begin{equation}\label{ineq:z:solo:Nminusdos}
    \zLB \leq C \bx_j^+ + D \bu_j^+ \leq \zUB, \;j\in\N_0^{N-2}.
\end{equation}
We now compute the value of $\bx^+_{N-1}$.
\begin{align} \label{eq:equivalence:bxN}
    \bx_{N-1}^+ &= A \bx_{N-2}^+ + B \bu_{N-2}^+ \nonumber \\
                &= A \bx_{N-1} + B \bu_{N-1} = \bx_N = \xe + \xc.
\end{align}
Since $\bu_{N-1}^+ = \ue + \uc$ we obtain
\begin{align*}
    z_{N-1}^+ &= C \bx^+_{N-1} + D \bu^+_{N-1} \\
              &= C (\xe + \xc) + D (\ue + \uc) = \ze + \zc.
\end{align*}

Defining
\begin{equation*}
    z_{N+\ell} = \ze + \zs \sin(w\ell) + \zc \cos(w\ell), \; \forall \ell \in \N,
\end{equation*}
we have $z_{N-1}^+ = z_{N}$. This, along with \eqref{ineq:z:bounds}, yields
\begin{equation*}
    \zLB \leq C\bx_{N-1}^++D\bu_{N-1}^+ \leq \zUB.
\end{equation*}
From this and \eqref{ineq:z:solo:Nminusdos}, we conclude
\begin{equation*}
    \zLB  \leq C\bx_j^++D\bu_j^+ \leq \zUB, \;j\in\N_0^{N-1},
\end{equation*}
which proves \eqref{ineq:Feas:z:first}.
The value of $\bx_N^+$ can be computed from $\bx_{N-1}^+ = \bx_N$ and $\bu_{N-1}^+ = \ue + \uc$ as follows.
\begin{align*}
    \bx_N^+  &= A \bx_{N-1}^+ + B \bu_{N-1}^+ = A \bx_N + B(\ue + \uc) \\
             &= A(\xe + \xc)+ B(\ue + \uc) = \xe^+ + \xc^+,
\end{align*}
which proves (\ref{eq:Feas:xN}).
From
\begin{equation*}
    \xe^+  = A \xe + B \ue = \xe,
\end{equation*}
and equality $\ue^+ = \ue$ (see \eqref{eq:Shift:ue}) we obtain from \eqref{eq:Shift:x_h}
\begin{equation*}
    \xe^+ = A \xe + B \ue = A \xe^+ + B \ue^+,
\end{equation*}
which proves \eqref{eq:Feas:xe}.
We now prove \eqref{eq:Feas:xs}.
\small
\begin{align*}
    A \xs^+ + B\us^+ &= A( A \xs + B\us) + B \us^+ \\
                               &=  A(\xs \cos(w) - \xc \sin(w)) \\
                               &\quad +  B(\us \cos(w) - \uc \sin(w)) \\
                               &=  (A \xs + B \us) \cos(w) - (A \xc + B \uc) \sin(w) \\
                               &=  \xs^+ \cos(w) - \xc^+ \sin(w).
\end{align*}
\normalsize
We prove \eqref{eq:Feas:xc} in a similar way.
\small
\begin{align*}
    A \xc^+ + B \uc^+ &= A( A \xc + B\uc) + B \uc^+ \\
                               &=  A(\xs \sin(w) + \xc \cos(w)) \\
                               &\quad +  B(\us \sin(w) + \uc \cos(w)) \\
                               &=  (A \xs + B \us) \sin(w) + (A \xc + B \uc) \cos(w) \\
                               &=  \xs^+ \sin(w) + \xc^+ \cos(w).
\end{align*}
\normalsize
Next, we express $\ze^+$, $\zs^+$ and $\zc^+$ in terms of $\ze$, $\zs$, $\zc$.
\small
\begin{align*}
\ze^+ &= C \xe^+ + D \ue^+ = C \xe + D\ue = \ze.\\
\zs^+ &= C \xs^+ + D \us^+ = C( A \xs + B \us ) + D \us^+ \\
           &= C( \xs \cos(w) - \xc \sin(w)) + D( \us \cos(w) - \uc \sin(w)) \\
           &= (C \xs + D \us) \cos(w) - (C \xc + D \uc) \sin(w)\\
           &= \zs \cos(w) - \zc \sin(w).\\
\zc^+ &= C \xc^+ + D \uc^+ = C( A \xc + B \uc ) + D \uc^+ \\
           &= C( \xs \sin(w) + \xc \cos(w)) + D( \us \sin(w) + \uc \cos(w)) \\
           &= (C \xs + D \us) \sin(w) + (C \xc + D \uc) \cos(w)\\
           &= \zs \sin(w) + \zc \cos(w).
\end{align*}
\normalsize
Therefore, for every $i\in\N_1^{n_z}$ we have
\begin{align*}
    &\zej^+ = \zej \\
    &\zsj^+ = \zsj \cos(w) - \zcj \sin(w) \\
    &\zcj^+ = \zsj \sin(w) + \zcj \cos(w).
\end{align*}
In view of Property \ref{prop:simple:armonics} this leads to
\begin{equation*}
    \sqrt{(\zsj^+)^2 + (\zcj^+)^2} = \sqrt{\zsj^2 + \zcj^2}, \; i\in\N_1^{n_z}.
\end{equation*}
From this we conclude that inequalities \eqref{ineq:Feas:z:minus} and \eqref{ineq:Feas:z:plus} are directly inferred from \eqref{ineq:HMPC:z:minus} and \eqref{ineq:HMPC:z:plus}.
\end{proof}

\subsection{Proof of the asymptotic stability of the HMPC controller} \label{app:proof:stability:HMPC}

The proof of Theorem \ref{theo:HMPC:Stability} relies on the following lemma, whose proof, as well as the proof of the theorem, are heavily influenced by the proofs of Theorem 1 and Lemma~1 from~\cite{Limon_TAC_2018}. However, in our case, we directly prove the existence of a function that satisfies the Lyapunov asymptotic stability conditions of Theorem \ref{theo:Lyapunov:Stability}.

\begin{lemma} \label{lemma:stability:x:is:optimal} 
    Consider a system \eqref{eq:Model} subject to \eqref{eq:Constraints} and assume that $N$ is greater or equal to its controllability index. Let $(x_r, u_r)$ be a given reference and $x$ be a state belonging to the feasibility region of the HMPC controller $\lH(x; x_r,u_r)$. Then, $x = \xhcero^*$ if and only if $x = \xeo$, where $\xhcero^*$ is given by \eqref{eq:optimal:harmonic:signals:x} and $\xeo$ by Lemma \ref{lemma:optimal:artificial:reference:HMPC}.
\end{lemma}

\begin{proof} 
Due to space considerations, we will drop the dependency w.r.t. $(x_r, u_r)$ from the notation of the functions. 
First, we show the implication $x = \xhcero^* {\rightarrow} x = \xeo$.~Assume that $x = \xhcero^*$ and let ${\offsetCostH^* \doteq \offsetCostH(\xH^*, \uH^*)}$ and ${\offsetCostH^\circ \doteq \offsetCostH(\xHo, \uHo)}$. Then, it can be shown that
$\lH^*(x; x_r, u_r) = \offsetCostH(\xH^*, \uH^*)$, i.e., that the optimal solution of $\lH(x; x_r, u_r)$ is given by
\begin{equation} \label{eq:lemma:stability:x:is:optimal:solution}
    x_j^* = \xhj^*, \quad u_j^* = \uhj^*, \quad \forall j\in\N_0^{N-1},
\end{equation}
where $\xhj^*$ and $\uhj^*$ are given by \eqref{eq:optimal:harmonic:signals}.

Indeed, the stage cost of \eqref{eq:lemma:stability:x:is:optimal:solution} is $\stageCostH(\vv{x}^*, \vv{u}^*, \xH^*, \uH^*) = 0$, which is its smallest possible value for all solutions in which $\xhcero^* = x$. Additionally, it can be shown that \eqref{eq:lemma:stability:x:is:optimal:solution} is a feasible solution of \eqref{eq:HMPC:dynamics}-\eqref{ineq:HMPC:z:plus}.

Next, we prove that $\offsetCostH^* = \offsetCostH^\circ$ by contradiction. Assume that $\offsetCostH^* > \offsetCostH^\circ$. Since $(\xHo, \uHo)$ is the unique minimizer of $\offsetCostH(\cdot)$ for all $(\xH, \uH)$ that satisfy \eqref{eq:HMPC:xe}-\eqref{ineq:HMPC:z:plus}, this implies that $(\xH^*, \uH^*) \neq (\xHo, \uHo)$.

Let $\hat{\vv{x}}_H$ be defined as
\begin{align*}
    \hat{\vv{x}}_H &= (\hat{x}_e, \hat{x}_s, \hat{x}_c) = \lambda \xH^* + (1 - \lambda) \xH^\circ \\
                   &= \lambda (\xe^*, \xs^*, \xc^*) + (1-\lambda) (\xeo, 0, 0), \; \lambda \in [0, 1],
\end{align*}
and $\hat{\vv{u}}_H$ similarly.
Then, since $N$ is greater or equal to the controllability index of the system, $\cc{Z}$ is convex, and $(\xhj^*, \uhj^*) \in \ri{\cc{Z}}$ for all $j\in\N$ (as can be deduced from Property \ref{prop:bounds} and the fact that $(\xH^*, \uH^*)$ satisfies \eqref{ineq:HMPC:z:minus} and \eqref{ineq:HMPC:z:plus}), there exists a $\hat{\lambda} \in [0, 1)$ such that for any $\lambda \in [\hat{\lambda}, 1]$ there is a dead-beat control law $\vv{u}^\text{db}$ for which the predicted trajectory $\vv{x}^\text{db}$ satisfying $x^\text{db}_0 = \xhcero^*$ and $x^\text{db}_N = \hat{x}_{h 0}$ is a feasible solution $(\vv{x}^\text{db}, \vv{u}^\text{db}, \hat{\vv{x}}_H, \hat{\vv{u}}_H)$ of problem $\lH(\xhcero^*; x_r, u_r)$.

Then, taking into account the optimality of \eqref{eq:lemma:stability:x:is:optimal:solution}, and noting that there exists a matrix $P\in\Sp{n}$ such that
$$\Sum{j=0}{N-1} \| x_j^\text{db} - \hat{x}_{h j} \|^2_Q + \| u_j^\text{db} - \hat{u}_{h j} \|^2_R \leq \|x_0^\text{db} - \hat{x}_{h 0} \|^2_P,$$
we have that
\begin{align} \label{eq:lemma:stability:offsetCostH}
    \offsetCostH^* &= H(\vv{x}^*, \vv{u}^*, \xH^*, \uH^*) \leq H(\vv{x}^\text{db}, \vv{u}^\text{db}, \hat{\vv{x}}_H, \hat{\vv{u}}_H) \nonumber \\
                               &= \stageCostH(\vv{x}^\text{db}, \vv{u}^\text{db}, \hat{\vv{x}}_H, \hat{\vv{u}}_H) + \offsetCostH(\hat{\vv{x}}_H, \hat{\vv{u}}_H) \nonumber \\
                               &\leq \| \xhcero^* - \hat{x}_{h0} \|^2_P + \offsetCostH(\hat{\vv{x}}_H, \hat{\vv{u}}_H) \nonumber \\
                               &\numeq{*} (1 - \lambda)^2 \| \xhcero^* - \xeo \|^2_P + \offsetCostH(\hat{\vv{x}}_H, \hat{\vv{u}}_H),
\end{align}
where step $(*)$ is using
\begin{align*}
    \xhcero^* - \hat{x}_{h0} &= \xhcero^* -\left[ \lambda\xhcero^* + (1-\lambda) \xhcero^\circ \right] \\
                             &= (1-\lambda) (\xhcero^* - \xhcero^\circ) = (1-\lambda) (\xhcero^* - \xeo).
\end{align*}
From the convexity of $\offsetCostH(\cdot)$ we have that
\begin{equation*}
    \offsetCostH(\hat{\vv{x}}_H, \hat{\vv{u}}_H) \leq \lambda \offsetCostH^* + (1-\lambda) \offsetCostH^\circ, \; \lambda \in [0, 1],
\end{equation*}
which combined with \eqref{eq:lemma:stability:offsetCostH} leads to,
\begin{equation} \label{eq:lemma:stability:x:is:optimal:upper:bound}
    \offsetCostH^* \leq \theta(\lambda), \, \lambda \in [\hat\lambda, 1],
\end{equation}
where
\begin{equation*}
    \theta(\lambda) \doteq (1 - \lambda)^2 \| \xhcero^* - \xeo \|^2_P + \lambda (\offsetCostH^* - \offsetCostH^\circ) + \offsetCostH^\circ.
\end{equation*}
The derivative of $\theta(\lambda)$ (w.r.t. $\lambda$) is
\begin{equation*}
    \theta'(\lambda) = -2 (1 - \lambda) \| \xhcero^* - \xeo ||^2_P + (\offsetCostH^* - \offsetCostH^\circ).
\end{equation*}
Taking into account the initial assumption $\offsetCostH^* - \offsetCostH^\circ > 0$, we have that $\theta'(1) > 0$. Therefore, there exists a $\lambda \in [\hat\lambda, 1)$ such that $\theta(\lambda) < \theta(1) = \offsetCostH^*$, which together with \eqref{eq:lemma:stability:x:is:optimal:upper:bound} leads to the contradiction $\offsetCostH^* < \offsetCostH^*$.

Therefore, we have that $\offsetCostH(\xH^*, \uH^*) \leq \offsetCostH(\xHo, \uHo)$. Moreover, since $(\xHo, \uHo)$ is the unique minimizer of $\offsetCostH(\xH, \uH)$ for all $(\xH, \uH)$ that satisfy \eqref{eq:HMPC:xe}-\eqref{ineq:HMPC:z:plus}, we conclude that ${\xhcero^* = \xeo}$.

The reverse implication is straightforward. Assume now that ${x = \xeo}$. Then,
\begin{equation} \label{eq:lemma:stability:x:is:optimal:solution:reverse}
    \xH^* = \xHo, \, \uH^* = \uHo, \, x_j^* = \xeo, \, u_j^* = \ueo, \, \forall j\in\N_0^{N-1}
\end{equation}
is a feasible solution of $\lH(x; x_r, u_r)$, since $(\xHo, \uHo)$ satisfies \eqref{eq:HMPC:xe}-\eqref{ineq:HMPC:z:plus} and $(\xeo, \ueo)$ is a steady state of the system. Moreover, \eqref{eq:lemma:stability:x:is:optimal:solution:reverse} is the optimal solution of $\lH(x; x_r, u_r)$. Indeed, note that $\stageCostH(\vv{x}^*, \vv{u}^*, \xH^*, \uH^*) = 0$ and that $\offsetCostH(\xH^*, \uH^*) = \offsetCostH(\xHo, \uHo)$, which is, once again, its minimum value for all $(\xH, \uH)$ satisfying \eqref{eq:HMPC:xe}-\eqref{ineq:HMPC:z:plus}. Therefore, due to the strict convexity of $\offsetCostH(\cdot)$, we conclude that $\xH^* = \xHo$, implying ${\xhcero^* = \xeo}$. \qedhere
\end{proof}

\begin{proof}[Proof of Theorem \ref{theo:HMPC:Stability}] 

The proof is based on finding a function that satisfies the Lyapunov conditions for asymptotic stability given in Theorem \ref{theo:Lyapunov:Stability}.

Let us consider a state $x$ belonging to the domain of attraction of the HMPC controller and a reference $(x_r, u_r)$. Let $\vv{x}^*$, $\vv{u}^*$, $\xH^*$, $\uH^*$ be the optimal solution of $\lH(x; x_r, u_r)$, $\lH^*(x; x_r, u_r) \doteq \costFunH(\vv{x}^*, \vv{u}^*, \xH^*, \uH^*)$ be its optimal value, $\offsetCostH^* \doteq \offsetCostH(\xH^*, \uH^*; x_r, u_r)$ and $\offsetCostH^\circ \doteq \offsetCostH(\xHo, \uHo; x_r, u_r)$.

We will now show that the function
\begin{equation*}
    W(x; x_r, u_r) = \lH^*(x; x_r, u_r) - \offsetCostH^\circ(x_r, u_r)
\end{equation*}
is a Lyapunov function for $x - \xeo$ by finding suitable $\alpha_1(\|x - \xeo\|)$ and $\alpha_2(\|x - \xeo\|)$ $\cc{K}_\infty$-class functions such that the conditions of Theorem \ref{theo:Lyapunov:Stability} are satisfied. Due to space considerations, in this proof we will drop the dependency w.r.t. $(x_r, u_r)$ from the notation of the functions.

Let $x^+ \doteq A x + B \bu_0^*$ be the successor state and consider the shifted sequence $\vv{x}^+$, $\vv{u}^+$, $\xH^+$, $\uH^+$ be defined as in \eqref{eq:Shift:all} but taking $\vv{x}^*$, $\vv{u}^*$, $\xH^*$, $\uH^*$ in the right-hand-side of the equations. It is clear from the proof of Theorem \ref{theo:Recursive:Feasibility} that this shifted sequence is a feasible solution of $\lH(x^+; x_r, u_r)$.

The satisfaction of condition \textit{(i)} of Theorem \ref{theo:Lyapunov:Stability}, for any state $x$ belonging to the domain of attraction of the HMPC controller, follows from
\begin{align*} \label{eq:stability:proof:condition:i}
    W(x) &= \Sum{j=0}{N-1} \| x_j^* - \xhj^* \|_Q^2 + \| u_j^* - \uhj^* \|_R^2 + \offsetCostH^* - \offsetCostH^\circ \nonumber \\
         &\numeq[\geq]{A1} \|{x_0^*}{-}{\xhcero^*} \|_Q^2 + \frac{\hat\sigma}{2} \|\xhcero^* - \xeo \|^2 \nonumber \\
         & \geq \min\{ \lambda_\text{min} (Q), \frac{\hat\sigma}{2}\} \left( \|x - \xhcero^*\|^2 + \|\xhcero^* - \xeo \|^2 \right) \nonumber \\
         &\numeq[\geq]{A2} \frac{1}{2} \min\{ \lambda_\text{min} (Q), \frac{\hat\sigma}{2} \} \| x - \xeo \|^2,
\end{align*}
where $(A2)$ is due to the parallelogram law, which states that for any two vectors $v_1, v_2 \inR{n_v}$,
\begin{equation*}
    \|v_1\|^2 + \|v_2\|^2 = \frac{1}{2}\|v_1 + v_2\|^2 + \frac{1}{2}\|v_1 - v_2\|^2,
\end{equation*}
and $(A1)$ follows from the fact that
\begin{equation} \label{eq:strong:convexity:condition:offsetCostH}
    \offsetCostH^* - \offsetCostH^\circ \geq \frac{\hat\sigma}{2} \|\xhcero^* - \xeo\|^2
\end{equation}
for some $\hat\sigma > 0$. To show this, note that $\offsetCostH(\cdot)$ is a strongly convex function. Therefore, it satisfies for some $\sigma > 0$ \cite[Theorem 5.24]{Beck_SIAM_17}, \cite[\S 9.1.2]{Boyd_ConvexOptimization},
\begin{equation*}
    \offsetCostH(z) - \offsetCostH(y) \geq \sp{\nabla \offsetCostH(y)}{z - y} + \frac{\sigma}{2} \| z - y \|^2,
\end{equation*}
for all $z, y \in \R^n \times \R^n \times \R^n \times \R^m \times \R^m \times \R^m$. Particularizing for $z = (\xH^*, \uH^*)$ and $y = (\xHo, \uHo)$ we have,
\begin{align*}
    \offsetCostH^* - \offsetCostH^\circ &\geq \sp{\nabla \offsetCostH^\circ}{(\xH^*, \uH^*) - (\xHo, \uHo)} \\
                                        &\quad+ \frac{\sigma}{2} \| (\xH^*, \uH^*) - (\xHo, \uHo) \|^2.
\end{align*}
From the optimality of $(\xHo, \uHo)$ we have that \cite[Proposition 5.4.7]{Bertsekas_Convex_2009}, \cite[\S 4.2.3]{Boyd_ConvexOptimization},
\begin{align*}
    \sp{\nabla \offsetCostH^\circ}{(\xH, \uH) - (\xHo, \uHo)} \geq 0
\end{align*}
for all $(\xH, \uH)$ satisfying \eqref{eq:HMPC:xe}-\eqref{ineq:HMPC:z:plus}. Since $(\xH^*, \uH^*)$ satisfies \eqref{eq:HMPC:xe}-\eqref{ineq:HMPC:z:plus}, this leads to
\begin{align*}
    &\offsetCostH^* - \offsetCostH^\circ \geq \frac{\sigma}{2} \| (\xH^*, \uH^*) - (\xHo, \uHo) \|^2 \\
                                        &= \frac{\sigma}{2} \left( \| \xH^* - \xHo \|^2 + \| \uH^* - \uHo \|^2 \right) \\
                                        &\geq \frac{\sigma}{2} ( \|{\xe^*} - {\xeo} \|^2 + \| \xs^* \|^2 + \| \xc^* \|^2 ) \\
                                        &\geq \frac{\sigma}{2} ( \|{\xe^*} - {\xeo} \|^2 + \| \xs^* \sin(- w N) \|^2 {+} \| \xc^* \cos(- w N)\|^2 ).
\end{align*}
Finally, making use of the parallelogram law, inequality \eqref{eq:strong:convexity:condition:offsetCostH} follows from the fact that there exists a scalar $\hat\sigma > 0$ such that
\begin{align*}
&\frac{\sigma}{2} ( \|{\xe^*} - {\xeo} \|^2 + \| \xs^* \sin(- w N) \|^2 + \| \xc^* \cos(- w N)\|^2 ) \\
&\geq \frac{\hat\sigma}{2} \|\xe^* - \xeo + \xs^* \sin(- w N) + \xc^* \cos(- w N) \|^2 \\
&= \frac{\hat\sigma}{2} \| \xhcero^* - \xeo \|^2.
\end{align*}

Since $(\xeo, \ueo) \in \ri{\cc{Z}}$ (see Lemma \ref{lemma:optimal:artificial:reference:HMPC}), the system is controllable and $N$ is greater than its controllability index, there exists a sufficiently small compact set containing the origin in its interior $\Omega$ such that, for all states $x$ that satisfy $x - \xeo \in \Omega$, the dead-beat control law
\begin{equation*} \label{eq:dead:beat:control:law}
    u_j^\text{db} = K_\text{db} (x_j^\text{db} - \xeo) + \ueo 
\end{equation*}
provides an admissible predicted trajectory $\vv{x}^\text{db}$ of system \eqref{eq:Model} subject to \eqref{eq:Constraints}, where $x^\text{db}_{j+1} = A x^\text{db}_j + B u^\text{db}_j,\; j \in\N_0^{N-1}$, ${x^\text{db}_0 = x}$ and ${x^\text{db}_N = \xeo}$.

Then, taking into account the optimality of $\vv{x}^*$, $\vv{u}^*$, $\xH^*$, $\uH^*$, we have that,
\begin{align*}
    W(x) &= \stageCostH(\vv{x}^*, \vv{u}^*, \xH^*, \uH^*) + \offsetCostH(\xH^*, \uH^*) - \offsetCostH^\circ \\
         &\leq \stageCostH(\vv{x}^\text{db}, \vv{u}^\text{db}, \xHo, \uHo) + \offsetCostH(\xHo, \uHo) - \offsetCostH^\circ \\
         &= \Sum{j=0}{N-1} \| x_j^\text{db} - \xeo \|^2_Q + \| u_j^\text{db} - \ueo \|^2_R.
\end{align*}
Therefore, there exists a matrix $P\in\Sp{n}$ such that
\begin{align*}
    W(x) \leq \lambda_\text{max}(P) \|x - \xeo\|^2
\end{align*}
for any $x - \xeo \in\Omega$. This shows the satisfaction of condition \textit{(ii)} of Theorem \ref{theo:Lyapunov:Stability}.

Next, let $\Delta W(x) \doteq W(x^+) - W(x)$ and note that, as shown by \eqref{eq:Shift:uN}, \eqref{eq:Shift:x_h}, \eqref{eq:equivalence:step:bxj_buj}, \eqref{eq:equivalence:bxN} and Property \ref{prop:simple:armonics}, we have that ${x_j^+ = x_{j+1}^*}$ for $j\in\N_0^{N-1}$, $u_j^+ = u_{j+1}^*$ for $j\in\N_0^{N-1}$, and that $\xhj^+ = x_{h,j+1}^*$ and $\uhj^+ = u_{h,j+1}^*$ for $j\in\N$.
Then,
\begin{align*}
    \Delta W(x) &= \lH^*(x^+) - \offsetCostH^\circ - \lH^*(x) + \offsetCostH^\circ \\
                &\leq \lH(x^+) - \lH^*(x) \\
                &= \Sum{j=0}{N-1} \| x_j^+ - \xhj^+ \|_Q^2 + \| u_j^+ - \uhj^+ \|_R^2\\
                & + \Sum{j=0}{N-1} - \| x_j^* - \xhj^* \|_Q^2 - \| u_j^* - \uhj^* \|_R^2 \\
                & + \offsetCostH(\xH^+, \uH^+) - \offsetCostH(\xH^*, \uH^*) \\
                & \numeq{*} \Sum{j=1}{N-1} \|x_j^* - \xhj^*\|_Q^2 + \|u_j^* - \uhj^*\|_R^2 \\
                & + \Sum{j=0}{N-1} - \|x_j^* - \xhj^*\|_Q^2 - \|u_j^* - \uhj^*\|_R^2 \\
                & + \| x_N^* - x_{hN}^* \|_Q^2 + \|u_N^* - u_{hN}^*\|_R^2 \\
                & = -\| x_0^* - \xhcero^* \|_Q^2 - \| u_0^* - \uhcero^* \|_R^2 \\
                & \leq - \lambda_\text{min}(Q) \| x - \xhcero^* \|^2,
\end{align*}
where in step $(*)$ we are making use of the fact that, 
$$\offsetCostH(\xH^+, \uH^+) = \offsetCostH(\xH^*, \uH^*).$$
Indeed, note that $\xe^+ = \xe^*$ and $\ue^+ = \ue^*$. Therefore, the first two terms of $\offsetCostH(\xH^+, \uH^+)$ \eqref{eq:HMPC:Offset:Cost} are the same as those of $\offsetCostH(\xH^*, \uH^*)$. We now show that, since $T_h$ and $S_h$ are diagonal matrices, the terms $\| \xs \|_{T_h}^2 + \| \xc \|_{T_h}^2$ are also the same (terms $\| \us \|_{S_h}^2 + \| \uc \|_{S_h}^2$ follow similarly).
\begin{align*}
\| \xs^+ \|_{T_h}^2 + \| \xc^+ \|_{T_h}^2 &= \| \xs^* \cos(w) - \xc^* \sin(w) \|_{T_h}^2 \\
                                          &+ \| \xs^* \sin(w) + \xc^* \cos(w) \|_{T_h}^2 \\
                                                &= ( \sin(w)^2 + \cos(w)^2 ) \| \xs^* \|_{T_h}^2 \\
                                                &+ ( \sin(w)^2 + \cos(w)^2 ) \| \xc^* \|_{T_h}^2 \\
                                                &+ 2 \cos(w) \sin(w) \sp{\xs^*}{T_h \xc^*} \\
                                                &- 2 \cos(w) \sin(w) \sp{\xs^*}{T_h \xc^*} \\
                                                &= \| \xs^* \|_{T_h}^2 + \| \xc^* \|_{T_h}^2.
\end{align*}

The satisfaction of condition \textit{(iii)} of Theorem \ref{theo:Lyapunov:Stability} now follows from noting that inequality
\begin{equation*}
    W(x_{k+1}) - W(x_k) \leq - \lambda_\text{min}(Q) \| x_k - \xhcero^*(x_k) \|^2
\end{equation*}
along with Lemma \ref{lemma:stability:x:is:optimal} leads to
\begin{align*}
    &W(x_{k+1}) < W(x_k), \; \forall x_k \neq \xeo, \\
    &W(x_{k+1}) = W(x_k), \; \text{if} \; x_k = \xeo. \qedhere
\end{align*}

\end{proof}

\end{appendix}

\newpage
\bibliographystyle{IEEEtran}
\bibliography{IEEEabrv,BibKrupa}

\end{document}